%% file: SQCLP_ICLP2010_TR.tex
\documentclass{tlp}

\pdfoutput=1

\usepackage{stmaryrd, amssymb, extarrows}
\usepackage{verbatim} 
\usepackage{hyperref}

\newtheorem{defn}{Definition}[section] 
\newtheorem{thm}{Theorem}[section] 
\newtheorem{lem}{Lemma}[section] 
\newtheorem{prop}{Proposition}[section] 
\newtheorem{exmp}{Example}[section] 
\newtheorem{cor}{Corollary}[section] 

\newcommand{\nc}{\newcommand}

\nc{\Con}[1]{\mbox{Con}_{#1}}  

\nc{\pc}[3]{\mathsf{#1}_{#2}(#3)} 
\nc{\qval}[1]{\pc{qVal}{}{#1}} 
\nc{\qbound}[1]{\pc{qBound}{}{#1}} 
\nc{\qvali}[2]{\pc{qVal}{#1}{#2}} 
\nc{\qboundi}[2]{\pc{qBound}{#1}{#2}} 

\nc{\encode}[1]{\ulcorner#1\urcorner} 

\nc{\spair}[2]{\llparenthesis\, #1, #2 \,\rrparenthesis}

\newcommand{\schemenp}[1]{\mbox{#1}} 
\newcommand{\scheme}[2]{\schemenp{#1}(#2)} 

\newcommand{\clp}[1]{\scheme{CLP}{#1}} 

\newcommand{\sqclp}[3]{\scheme{SQCLP}{#1,#2,#3}} 

\newcommand{\cdom}{\mathcal{C}} 

\newcommand{\rdom}{\mathcal{R}} 
\newcommand{\hdom}{\mathcal{H}} 
\newcommand{\fdom}{\mathcal{FD}} 

\newcommand{\qdom}{\mathcal{D}} 
\newcommand{\aqdomd}[1]{D_{#1} \setminus \{\bt\}} 
\newcommand{\bqdomd}[1]{(D_{#1} \setminus \{\bt\}) \uplus \{?\}} 
\newcommand{\aqdom}{\aqdomd{}} 
\newcommand{\bqdom}{\bqdomd{}} 

\newcommand{\bt}{\mathrm{\mathbf{b}}} 
\newcommand{\tp}{\mathrm{\mathbf{t}}} 
\newcommand{\dleq}{\trianglelefteqslant} 
\newcommand{\dlt}{\vartriangleleft} 
\newcommand{\dgeq}{\trianglerighteqslant} 

\newcommand{\B}{\mathcal{B}} 
\newcommand{\U}{\mathcal{U}} 
\newcommand{\W}{\mathcal{W}} 
\newcommand{\UW}{\U\!\times\!\W} 

\newcommand{\simrel}{\mathcal{S}} 
\newcommand{\sid}{\simrel_{\mathrm{id}}} 

\newcommand{\toy}{\mathcal{TOY}}

\newcommand{\Prog}{\mathcal{P}} 

\newcommand{\Var}{\mathcal V\!ar} 
\newcommand{\War}{\mathcal W\!ar} 
\newcommand{\X}{\mathcal{X}} 

\newcommand{\set}[2]{\mathrm{#1}(#2)} 
\newcommand{\domset}[1]{\set{dom}{#1}} 
\newcommand{\ranset}[1]{\set{vran}{#1}} 
\newcommand{\varset}[1]{\set{var}{#1}} 
\newcommand{\warset}[1]{\set{war}{#1}} 

\nc{\Exp}{\mbox{Exp}_{\bot}(\Sigma,B,\Var)} 
\nc{\TExp}{\mbox{Exp}(\Sigma,B,\Var)} 
\nc{\GExp}{\mbox{Exp}_{\bot}(\Sigma,B)} 
\nc{\TGExp}{\mbox{Exp}(\Sigma,B)} 
\nc{\Term}{\mbox{Term}_{\bot}(\Sigma,B,\Var)} 
\nc{\TTerm}{\mbox{Term}(\Sigma,B,\Var)} 
\nc{\GTerm}{\mbox{Term}_{\bot}(\Sigma,B)} 
\nc{\TGTerm}{\mbox{Term}(\Sigma,B)} 

\nc{\At}{\mbox{At}(\Sigma,B,\Var)} 
\nc{\GAt}{\mbox{GAt}(\Sigma,B)} 
\nc{\PAt}{\mbox{PAt}(\Sigma,B,\Var)} 
\nc{\GPAt}{\mbox{GPAt}(\Sigma,B)} 

\nc{\Atz}{\mbox{At}_{\Sigma}} 
\nc{\QAtz}{\mbox{At}_{\Sigma}(\qdom)} 

\nc{\sust}{\mbox{Subst}_\Sigma} 
\nc{\Sust}{\mbox{Subst}(\Sigma,B,\Var)} 
\nc{\GSust}{\mbox{GSubst}(\Sigma,B)} 

\nc{\Soln}[3]{\mbox{Sol}_{#1}^{#2}(#3)} 
\nc{\Sol}[2]{\Soln{#1}{}{#2}} 
\nc{\GSol}[2]{\mbox{GSol}_{#1}(#2)}
\nc{\Solc}[1]{\Sol{\cdom}{#1}}
\nc{\CAns}[2]{\mbox{C\!Ans}_{#1}(#2)} 

\newcommand{\qat}[2]{#1 \sharp #2} 
\newcommand{\cqat}[3]{\qat{#1}{#2} \Leftarrow #3} 

\newcommand{\qgets}[1]{\xleftarrow{#1}} 

\newcommand{\Int}[1]{\mbox{Int}_{#1}}
\newcommand{\Intdc}{\Int{\qdom,\cdom}} 

\newcommand{\ibot}{\bot\!\!\!\bot} 
\newcommand{\itop}{\top\!\!\!\top} 

\newcommand{\I}{\mathcal{I}} 
\newcommand{\J}{\mathcal{J}} 

\nc{\closure}[1]{\mbox{cl}_{#1}} 

\newcommand{\Tp}{\mbox{T}_{\!\Prog}} 

\newcommand{\model}[1]{~{\models_{#1}}~} 

\newcommand{\M}[1]{\mathcal{M}_{#1}} 
\newcommand{\Mp}{\M{\Prog}} 

\newcommand{\diff}{~{\Longleftrightarrow_{\mathrm{def}}}~} 
\newcommand{\eqdef}{~{=_{\mathrm{def}}}~} 
\newcommand{\union}{\bigcup} 
\newcommand{\inter}{\bigcap} 
\newcommand{\supr}{\bigsqcup} 
\newcommand{\infi}{\bigsqcap} 
\newcommand{\entail}[1]{~{\succcurlyeq_{#1}}~} 

\newcommand{\sep}{\talloblong} 

\newcommand{\NAT}{\mathbb{N}}
\newcommand{\REAL}{\mathbb{R}}

\newcommand{\tup}[1]{\overline{#1}}   
\newcommand{\ntup}[2]{\tup{#1}_{#2}}  

\newcommand{\infx}[2]{\ {\vdash}_{\!#1}^{\!#2}\ } 
\newcommand{\infxx}[2]{\ {\vdash\!\!\vdash}_{\!#1}^{\!#2}\ } 

\newcommand{\SQCHL}{\mbox{SQCHL}} 
\newcommand{\sqchln}[4]{\infx{#1,#2,#3}{#4}} 
\newcommand{\sqchl}[3]{\sqchln{#1}{#2}{#3}{}} 
\newcommand{\sqchlrdc}{\sqchln{\simrel}{\qdom}{\cdom}{}} 
\newcommand{\sqchlrdcn}[1]{\sqchln{\simrel}{\qdom}{\cdom}{#1}} 

\newcommand{\isqchln}[4]{\infxx{#1,#2,#3}{#4}} 
\newcommand{\isqchlrdc}{\isqchln{\simrel}{\qdom}{\cdom}{}} 

\begin{document}

\long\def\comment#1{}

\title[Fixpoint \& Proof-theoretic Semantics for SQCLP]
{Fixpoint \& Proof-theoretic Semantics\\ for CLP with Qualification and Proximity \thanks{This work has been partially supported by the Spanish projects STAMP (TIN2008-06622-C03-01), PROMETIDOS--CM (S2009TIC-1465) and GPD--UCM (UCM--BSCH--GR58/08-910502).}\\
{\large Technical Report SIC-1-10}}

\author[M. Rodr\'iguez-Artalejo and C. A. Romero-D\'iaz]
{MARIO RODR\'IGUEZ-ARTALEJO and CARLOS A. ROMERO-D\'IAZ \\
Departamento de Sistemas Inform\'aticos y Computaci\'on, Universidad Complutense\\
Facultad de Inform\'atica, 28040 Madrid, Spain\\
\email{mario@sip.ucm.es, cromdia@fdi.ucm.es}
}

\pagerange{\pageref{firstpage}--\pageref{lastpage}}
\volume{\textbf{10} (3):}
\jdate{March 2002}
\setcounter{page}{1}
\pubyear{2002}

\maketitle

\thispagestyle{empty}


\begin{abstract}
Uncertainty in Logic Programming has been investigated during the last decades,
dealing with various extensions of the classical LP paradigm and different applications.
Existing proposals rely on different approaches, such as 
clause annotations based on uncertain truth values,
qualification values as a generalization of uncertain truth values,
and unification based on proximity  relations.
On the other hand,  the CLP scheme 
has established itself as a powerful extension of LP
that supports efficient computation over specialized domains 
while keeping a  clean declarative semantics.
In this report we propose a new scheme SQCLP
designed as an extension of CLP that supports qualification values and proximity relations.
We show that several  previous proposals 
can be viewed as particular cases of the new scheme, obtained by partial instantiation.
We present a declarative semantics for SQCLP that is based on observables, 
providing fixpoint and proof-theoretical characterizations of least program models
as well as an implementation-independent  notion of goal solutions.
\end{abstract}


\input{J1} 

\input{J2_0}
\input{J2_1}
\input{J2_2}
\input{J2_3}

\input{J3_0}

\input{J3_1}
\input{J3_2}
\input{J4} 

\section*{Acknowledgements}

This report is a widely extended version of \cite{RR10}.
The authors are thankful to the anonymous referees of \cite{RR10} for constructive remarks and suggestions which helped to improve the presentation.
They are also thankful to Rafael Caballero for useful discussions on the report's topics and to Jes\'us Almendros for pointing to bibliographic references in the area of flexible query answering.

\bibliographystyle{acmtrans}
\bibliography{../biblio}

\end{document}

%% file: J1.tex
\section{Introduction}
\label{sec:introduction}


Many extensions of logic programming (shortly LP) to deal with uncertainty have been proposed in the last decades.
A line of research not related to this report is based on probabilistic extensions of LP such as \cite{NS92}.
Other proposals in the field replace classical two-valued logic 
by some kind of many-valued logic whose truth values 
can be attached to computed answers and are usually  interpreted as certainty degrees.
The next paragraphs summarize some relevant approaches of this kind. 

There are extensions of LP using annotations in program clauses
to compute a certainty degree for the head atom from the certainty degrees previously computed for the body atoms.
This line of research includes the seminal proposal of Quantitative Logic Programming  by \cite{VE86} 
and inspired later works such as the Generalized Annotated logic Programs (shortly GAP) by  \cite{KS92}
and the QLP scheme for Qualified LP \cite{RR08}.
While \cite{VE86} and other early approaches used real numbers of the interval $[0,1]$ as certainty degrees,
QLP and GAP take elements from a parametrically given lattice to be used in annotations and attached to computed answers. 
In the case of QLP, the lattice is called a {\em qualification domain} 
and its elements (called {\em qualification values})  are not always understood as certainty degrees.
As argued in \cite{RR08}, GAP is a more general framework, 
but QLP's semantics  have some advantages for its intended scope.

There are also extended LP languages based on fuzzy logic \cite{Zad65,Haj98}, 
which can be classified into two major lines.
The first line includes Fuzzy LP languages such as \cite{Voj01,VGM02,GMV04} 
and the Multi-Adjoint LP (shortly MALP) framework by \cite{MOV01a,MOV01b}.
All these approaches extend classical LP by using clause annotations and a fuzzy interpretation of the connectives
and aggregation operators occurring in program clauses and goals. 
There is a relationship between Fuzzy LP and GAP that has been investigated in \cite{KLV04}.
Intended applications of Fuzzy LP languages include expert knowledge representation. 

The second line includes Similarity-based LP (shortly SLP)
 in the sense of \cite{AF99,Ses02,LSS04} and related proposals,
which keep the classical syntax of LP clauses but use a {\em similarity relation} over a set of symbols $S$ 
to allow ``flexible'' unification of syntactically different symbols with a certain approximation degree.
Similarity relations over a given set $S$ have been defined in \cite{Zad71,Ses02} and related literature as 
fuzzy relations represented by mappings $\simrel : S \times S \to [0,1]$ which satisfy reflexivity, symmetry and transitivity  axioms analogous to those required for classical equivalence relations. A more general notion called {\em proximity relation} was 
introduced in \cite{DP80} by omitting the transitivity axiom. 
As  noted by  \cite{SM99} and other authors, the transitivity property required for similarity relations
may conflict with user's intentions in some cases. 
The {\sf Bousi$\sim$Prolog} language \cite{JRG09,JR09b,JR09} 
has been designed with the aim of generalizing SLP to work with proximity relations.
A different generalization of SLP is the SQLP scheme \cite{CRR08},
designed as an extension of the QLP scheme.
In addition to clause annotations in QLP style, SQLP uses a given similarity 
relation $\simrel : S \times S \to D$ (where $D$ is the carrier set of a parametrically given qualification domain)
in order to support flexible unification.
In the sequel we use the acronym SLP as including proximity-based LP languages also.
Intended applications of SLP include flexible query answering.  
An analogy of proximity relations in a different context (namely partial constraint satisfaction) can be found in \cite{FW92}, where several metrics are proposed to measure the proximity between the solution sets of two different constraint satisfaction problems.

Several of the above mentioned LP extensions (including GAP, QLP, the Fuzzy LP language in  \cite{GMV04} and SQLP) have used constraint solving as an implementation technique.
However, we only know two approaches which have been conceived as extensions of the classical CLP scheme \cite{JL87}.
Firstly, \cite{Rie96,Rie98phd} extended the formulation of CLP by \cite{HS88} with quantitative LP in the sense of \cite{VE86}; this work was motivated by problems from the field of natural language processing. 
Secondly, \cite{BMR01} proposed a semiring-based approach to CLP, where constraints are solved in a soft way with levels of consistency represented by values of a semiring.
This approach was motivated by constraint satisfaction problems and implemented with {\tt clp(FD,S)} in \cite{GC98} for a particular class of semirings which enable to use local consistency algorithms.
The relationship between \cite{Rie96,Rie98phd,BMR01} and the results of this report will be further discussed in Section \ref{sec:conclusions}.

Finally, there are a few preliminary attempts to combine some of the above mentioned approaches 
with the Functional Logic Programming (shortly FLP) paradigm 
found in languages such as {\sf Curry} \cite{curry} and $\toy$ \cite{toy}. 
Similarity-based unification for FLP languages has been investigated by \cite{MP07}, 
while \cite{CRR09} have proposed a generic scheme QCFLP designed as a common extension of
the two schemes CLP and QLP with first-order FLP features.


In this report we propose a new extension of CLP that supports qualification values and proximity relations.
More precisely, we define a generic scheme SQCLP whose instances $\sqclp{\simrel}{\qdom}{\cdom}$ are
parameterized by a proximity relation $\simrel$, a qualification domain $\qdom$ and a constraint domain $\cdom$. 
We will show that several previous proposals 
can be viewed as particular cases of SQCLP, obtained by partial instantiation.
Moreover, we will present a declarative semantics for SQCLP that is inspired in the observable CLP semantics by \cite{GL91,GDL95}
and provides fixpoint and proof-theoretical characterizations of least program models
as well as an implementation-independent notion of goal solution that can be used to specify the expected behavior of goal solving systems.

The reader is assumed to be familiar with the semantic foundations of LP \cite{Llo87,Apt90} and CLP \cite{JL87,JMM+98}.
The rest of the report is structured as follows:
Section \ref{sec:domains} introduces constraint domains, qualification domains and proximity relations.
Section \ref{sec:sqclp} presents the SQCLP scheme and the main results on its declarative semantics.
Finally, Section \ref{sec:conclusions} concludes by giving an overview of related approaches (many of which can be viewed as particular cases of SQCLP) and pointing to some lines open for future work. 

%% file: J2_0.tex
\section{Constraints, Qualification \& Proximity}
\label{sec:domains} 

%% file: J2_1.tex
\subsection{Constraint Domains}
\label{sec:cdoms}


The Constraint Logic Programming paradigm (CLP) was introduced in \cite{JL87}
with the aim of generalizing the Herbrand Universe which underlies classical Logic Programming (LP) 
to other domains tailored to specific application areas.
In this seminal paper, CLP was introduced as a generic scheme with instances $\clp{\cdom}$
parameterized by {\em constraint domains} $\cdom$, each of which supplies several items: 
a {\em constraint language} providing a class of domain specific formulae, 
called {\em constraints} and serving as logical conditions in $\clp{\cdom}$ programs and computations;
a {\em constraint structure} serving as interpretation of the constraint language; 
a {\em constraint  theory} serving as a basis for proof-theoretical deduction with constraints; 
and a {\em constraint solver} for checking constraint satisfiability. 
Certain assumptions were made to ensure the proper relationship
between the constraint language, structure, theory and solver, so that the
classical results on the operational and declarative semantics of LP \cite{Llo87,Apt90} 
could be extended to all the $\clp{\cdom}$ languages. 
A revised and updated presentation of the main results from \cite{JL87}  can be found in \cite{JMM+98},
while a survey of CLP as a programming paradigm is given in \cite{JM94}.


The notion of constraint domain is a key ingredient of the CLP scheme. 
In addition to the classical formulation in \cite{JL87,JMM+98}, other formalizations have been used for different purposes. 
Some significative examples are: 
the CLP scheme proposed in \cite{HS88}, motivated by applications to computational linguistics
and allowing more than one constraint structure to come along with a given constraint language;
the proof-theoretical notion of constraint system given in \cite{Sar92}, intended for application to concurrent constraint languages;
and the constraint systems proposed in \cite{LOPP08} as the basis of a functorial semantics for CLP with negation, where a single constraint structure is replaced by a class of elementary equivalent structures.

 
In this paper we will use a simple notion of constraint domain, motivated by three main considerations: 
firstly, to focus on declarative semantics, rather than proof-theoretic or operational issues;
secondly, to provide a purely relational framework;
and thirdly, to clarify the interplay between domain-specific programming resources such as  basic values  and primitive predicates,
and general-purpose programming resources such as data constructors and defined predicates.
 
\subsubsection{Preliminary notions}
\label{sec:cdoms:syntax}

Before presenting constraint domains in a formal way, 
let us introduce some mainly syntactic notions that will be used all along the paper. 


\begin{defn}[Signatures]
\label{def:sig}
We assume a {\em universal programming signature} $\ \Gamma = \langle DC, DP \rangle$  where $DC = \bigcup_{n \in \NAT}DC^n$ and $DP \!= \bigcup_{n \in \NAT}  DP^n$ are infinite and mutually disjoint sets of free function symbols (called {\em data constructors} in the sequel) and {\em defined predicate} symbols, respectively, ranked by arities.
We will use {\em domain specific signatures} $\Sigma = \langle DC, DP, PP \rangle$
extending $\Gamma$ with a disjoint set $PP = \bigcup_{n \in \NAT}  PP^n$ of {\em primitive predicate} symbols, also ranked by arities.
The idea is that primitive predicates come along with constraint domains, while defined predicates are specified in user programs.
Each $PP^n$ maybe any countable set of $n$-ary predicate symbols.
In practice, $PP$ is expected to be a finite set. \mathproofbox
\end{defn}


In the sequel, we assume that any signature $\Sigma$ includes two nullary constructors 
{\sf true}, {\sf false} $\in DC^0$ to represent the boolean values, 
a binary constructor {\sf pair} $\in DC^2$ to represent ordered pairs, 
as well as constructors to represent lists and other common data structures. 
Given a signature $\Sigma$, a set $B$ of {\em basic values} $u$ and a countably infinite set $\Var$ of variables $X$, 
{\em terms} and {\em atoms} are built as defined below, where $\ntup{o}{n}$ abbreviates the $n$-tuple of syntactic objects $o_1, \ldots, o_n$ and $\varset{o}$ denotes the set of all variables occurring in the syntactic object $o$.

\begin{defn}[Terms and atoms]
\label{def:ta}
\begin{itemize}
\item
{\em Constructor Terms}  $t \in \TTerm$ have the syntax $t ::= X | u | c(\ntup{t}{n})$, where $c \in DC^n$. 
They will be called just terms in the sequel. 
In concrete examples, we will use Prolog syntax for terms built with list constructors,
and we will write $(t_1,t_2)$ rather than {\sf pair}$(t_1,t_2)$ for terms representing ordered pairs.
\item 
The set of all the variables occurring in $t$ is noted as $\varset{t}$.
A term $t$  is called {\em ground} iff $\varset{t} = \emptyset$. 
$\TGTerm$ stands for the set of all ground terms.
\item
{\em Atoms} $A \in \At$ can be 
{\em defined atoms} $r(\ntup{t}{n})$, where $r \in DP^n$ and $t_i \in \TTerm$ ($1 \leq i \leq n$);
{\em primitive atoms} $p(\ntup{t}{n})$, where $p \in PP^n$ and $t_i \in \TTerm$ ($1 \leq i \leq n$);
and {\em equations} $t_1$ {\sf ==} $t_2$, where $t_1, t_2 \in \TTerm$
and `==' is the {\em equality symbol}, which does not belong to the signature $\Sigma$.
Primitive atoms are noted as $\kappa$ and the set of all primitive atoms is noted $\PAt$.
Equations and primitive atoms are collectivelly called {\em $\cdom$-based atoms}.
\item
The set of all the variables occurring in $A$ is noted as $\varset{A}$.
An atom $A$ is called {\em ground} iff $\varset{A} = \emptyset$. 
The set of all ground atoms (resp. ground primitive atoms) is noted as $\GAt$ (resp. $\GPAt$). \mathproofbox
\end{itemize}
\end{defn}


Note that the equality symbol `==' used as part of the syntax of equational atoms is not the same as the symbol `=' generally used for mathematical equality.  In particular, metalevel equations $o = o'$ can be used to assert the identity of two syntactical objects $o$ and $o'$.


Following well-known ideas, the syntactical structure of terms and atoms can be represented by means of trees with nodes
labeled by signature symbols, basic values and variables. 
In the sequel  we will use the notation $\Vert t \Vert$ to denote the {\em syntactical size} of $t$ 
measured as the number of nodes in the tree representation of $t$.
The {\em positions} of nodes in this tree can be noted as finite  sequences $p$ of natural numbers. 
In particular, the empty sequence $\varepsilon$ represents the root position.
The next definition presents essential notions concerning positions in terms.
Positions in atoms can be treated similarly.

\begin{defn}[Positions]
\label{def:pos}
\begin{enumerate}
\item
The set $\mbox{pos}(t)$ of positions of the term $t$ is defined by recursion on the structure of $t$:
	\begin{itemize}
	\item
	$\mbox{pos}(X) = \{\varepsilon\}$ for each variable $X \in \Var$.
	\item
	$\mbox{pos}(u) = \{\varepsilon\}$ for each basic value $u \in B$.
	\item
	$\mbox{pos}(c(\ntup{t}{n})) = \{\varepsilon\} \cup \bigcup_{i = 1}^n \{iq \mid q \in \mbox{pos}(t_i)\}$ for each $c \in DC^n$.
	\end{itemize}
\item
Given $p \in \mbox{pos}(t)$, the symbol $t \circ p$ of  $t$ at position $p$ is defined recursively:
	\begin{itemize}
	\item
	$X \circ \varepsilon = X$ for each variable $X \in \Var$.
	\item
	$u \circ \varepsilon  = u$ for each basic value $u \in B$.
	\item
	$c(t_1, \ldots, t_n) \circ \varepsilon  = c$ if $c \in DC^n$.
	\item
	$c(t_1, \ldots, t_n) \circ iq = t_i \circ q$ if $c \in DC^n$, $1 \leq i \leq n$ and $q \in \mbox{vpos}(t_i)$.
	\end{itemize}
\item
Given $p \in \mbox{pos}(t)$, the subterm $t |_p$ of $t$ at position $p$ is defined as follows:
	\begin{itemize}
	\item
	$t |_\varepsilon = t$ for any $t$.
	\item
	$c(t_1, \ldots, t_n) | _{iq} = t_i | _{q}$ if $c \in DC^n$, $1 \leq i \leq n$ and $q \in \mbox{pos}(t_i)$.
	\end{itemize}
\item
$p \in \mbox{pos}(t)$ is called a {\em variable position} of $t$  iff $t | p$ is a variable, and a {\em rigid position} of $t$ otherwise.
We define $\mbox{vpos}(t) = \{p \in \mbox{pos}(t) \mid p \mbox{ is a variable position}\}$ and  	
$\mbox{rpos}(t) = \{p \in \mbox{pos}(t) \mid p \mbox{ is a rigid position}\}$.
\item
Given $p \in \mbox{vpos}(t)$ and another term $s$, the result of replacing $s$ for the subterm of $t$ at position $p$ is 
noted as $t[s]_p$. See e.g. \cite{BN98} for a recursive definition. \mathproofbox
\end{enumerate}
\end{defn}


As usual, {\em substitutions}  are defined as mappings $\sigma : \Var \to \TTerm$ assigning terms to variables.
 The set of all substitutions is noted as $\Sust$. 
 Substitutions are extended to act over terms and other syntactic objects $o$ in the natural way. 
 By convention, the result of replacing each variable $X$ occurring in $o$ by $\sigma(X)$ is noted as $o\sigma$.
 Other common notions concerning substitutions are defined as follows:

\begin{defn}[Notions concerning Substitutions]
\label{def:sub} 
 \begin{itemize}
 \item 
 The {\em composition} $\sigma\sigma'$ of two substitutions is such that $o(\sigma\sigma')$ equals $(o\sigma)\sigma'$.
  \item
  For a given $\sigma \in \Sust$, the {\em domain} $\domset{\sigma}$ is defined as $\{X \in \Var \mid X\sigma \neq X\}$,
  and the {\em variable range} $\ranset{\sigma}$ is defined as $\bigcup_{X \in \domset{\sigma}} \varset{X\sigma}$.
  \item
  A substitution $\sigma$ is called {\em ground} iff $X \sigma$ is a ground term for all $X \in \domset{\sigma}$.
  The set of all ground substitutions is noted $\GSust$.
  \item
  A substitution $\sigma$ is called {\em finite} iff $\domset{\sigma}$ is a finite set, say $\{X_1, \ldots, X_k\}$.
  In this case, $\sigma$ can be represented as the {\em set of bindings} $\{X_1 \mapsto t_1, \ldots, X_k \mapsto t_k\}$,
  where $t_i = X_i\sigma$ for all $1 \leq i \leq k$. 
  \item
  Assume two substitutions $\sigma$, $\sigma'$, a set of variables $\X$ and a variable $Y$. 
  The notation $\sigma =_{\X} \sigma'$ means that $X\sigma = X\sigma'$ holds for all variables $X \in \X$.
  We also write $\sigma =_{\setminus \X} \sigma'$ and $\sigma =_{\setminus Y} \sigma'$ to abbreviate
  $\sigma =_{\Var \setminus \X} \sigma'$ and $\sigma =_{\Var \setminus \{Y\}} \sigma'$, respectively. \mathproofbox
  \end{itemize}
\end{defn}

\subsubsection{Constraint domains, constraints and their solutions}
\label{sec:cdoms:domains}

We are now prepared to present constraint  domains as mathematical structures providing a set of basic values along with an terms and an interpretation of primitive predicates\footnote{As we will see in Section \ref{sec:sqclp}, the interpretation of defined predicate symbols is program dependent.}.
The formal  definition is as follows:

\begin{defn}[Constraint Domains]
\label{def:cd}
A {\em Constraint Domain} of signature $\Sigma$ is any relational structure of the form 
$\cdom = \langle C, \{p^\cdom \mid p \in PP\}\rangle$ such that:
\begin{enumerate}
\item 
The carrier set $C$ is $\TGTerm$ for a certain set $B$ of {\em basic values}. 
When convenient, we note $B$ and $C$ as $B_\cdom$ and $C_\cdom$, respectively.
\item 
$p^{\cdom} : C^n \to \{0,1\}$, written simply as $p^{\cdom} \in \{0,1\}$ in the case $n = 0$, 
is called the {\em interpretation} of $p$ in $\cdom$. 
A ground primitive atom $p(\ntup{t}{n})$ is {\em true} in $\cdom$ iff $p^{\cdom}(\ntup{t}{n})=1$;
otherwise $p(\ntup{t}{n})$ is {\em false} in $\cdom$.  \mathproofbox
\end{enumerate}
\end{defn}
 

For the examples in this paper we will use a constraint domain $\rdom$ which allows to work with arithmetic constraints over the real numbers, as formalized in Definition \ref{def:rdom} below.

\begin{defn}[The Real Constraint Domain $\rdom$]
\label{def:rdom}
The constraint  domain $\rdom$ is defined to include:
\begin{itemize}
\item
The set of basic values $B_\rdom = \REAL$. 
Note that $C_\rdom$ includes ground terms built from real values and data constructors, in addition to real numbers.
\item
Primitive predicates for encoding the usual arithmetic operations over $\REAL$.
For instance, the addition operation $+$ over $\REAL$  is encoded by a ternary primitive predicate $op_+$ 
such that, for any $t_1, t_2 \in C_\rdom$, $op_{+}(t_1,t_2,t)$  is true in $\rdom$ iff $t_1, t_2, t \in \REAL$ and $t_1 + t_2 = t$.
In particular, $op_{+}(t_1,t_2,t)$ is false in $\rdom$ if either $t_1$ or $t_2$ includes data constructors.
The primitive predicates encoding other arithmetic operations such as $\times$ and $-$ are defined analogously.
\item
Primitive predicates for encoding the usual inequality  relations over $\REAL$.
For instance, the ordering $\leq$ over $\REAL$  is encoded by a binary primitive predicate $cp_{\leq}$
such that, for any $t_1, t_2 \in C_\rdom$,  $cp_{\leq}(t_1,t_2)$ is true in $\rdom$ iff $t_1, t_2, t \in \REAL$ and $t_1 \leq t_2$.
In particular,  $cp_{\leq}(t_1,t_2)$ is false in $\rdom$ if either $t_1$ or $t_2$ includes data constructors.
The primitive predicates encoding the other inequality relations, namely $>$, $\geq$ and $>$, are defined analogously. \mathproofbox
\end{itemize} 
\end{defn}

The domain $\rdom$ is well known as the basis of the $\clp{\rdom}$ language and system \cite{JMSY92}. 
Some presentations of $\rdom$ known in the literature represent the arithmetical operations by using  primitive functions instead of primitive predicates. In this paper we have chosen to work in a purely relational framework in order to simplify some technicalities without loss of real expressivity. 

Other useful instances of constraint domains are known in the Constraint Programming  literature; see e.g. \cite{JM94,LRV07}.
In particular, the {\em Herbrand} domain $\mathcal{H}$ is intended to work just with equality constraints, while $\mathcal{FD}$ allows to work with constraints involving {\em finite domain variables}.
The set of basic values of $\mathcal{FD}$ is $\mathbb{Z}$.
There are also known techniques for combining several given constraint domains into a more expressive one; see e.g. the {\em coordination domains} defined in \cite{EHR+09}.

The following definition introduces constraints over a given domain:


\begin{defn}[Constraints and Their Solutions]
\label{def:consol}
Given a constraint domain $\cdom$ of signature $\Sigma$:
\begin{enumerate}
\item
{\em Atomic constraints} over $\cdom$ are of two kinds: 
primitive atoms $p(\ntup{t}{n})$ and equations $t_1$ {\sf ==} $t_2$.
\item
 {\em Compound constraints} are built from atomic constraints using logical conjunction $\land$, 
existential quantification $\exists$, and sometimes other logical operations.
Constraints of the form $\exists X_1 \ldots \exists X_n(B_1 \land \ldots \land B_m)$ --where $B_j ~ (1 \leq j \leq m)$ are atomic--
are called {\em existential}.
The set of all constraints over $\cdom$ is noted $\Con{\cdom}$.
 \item
 Substitutions $\sigma : \Var \to \TTerm$ where $\TTerm$ is built using the set $B_\cdom$ of basic values
 are called $\cdom$-{\em substitutions}. 
 Ground substitutions $\eta \in \GSust$ are called {\em variable valuations}.
 The set of all possible variable  valuations  is noted $\mbox{Val}_\cdom$.
\item
{\em The solution set} $\Solc{\pi}$ of a constraint $\pi \in \Con{\cdom}$ is defined by recursion on $\pi$'s syntactic structure as follows:
    \begin{itemize}
    \item
    If $\pi$ is a primitive atom $p(\ntup{t}{n})$, 
    then $\Solc{\pi}$ is the set of all $\eta \in \mbox{Val}_\cdom$ such that 
    $p(\ntup{t}{n})\eta$ is ground and true in $\cdom$.
    \item
    If $\pi$ is an equation $t_1$ {\sf ==} $t_2$,
    then $\Solc{\pi}$ is the set of all $\eta \in \mbox{Val}_\cdom$ such that 
    $t_1\eta$ and $t_2\eta$ are ground and syntactically identical terms.
    \item
    If $\pi$ is $\pi_1 \land \pi_2$ then $\Solc{\pi} = \Solc{\pi_1} \cap \Solc{\pi_2}$.
    \item
    If $\pi$ is $\exists X \pi'$ then $\Solc{\pi}$ is the set of all $\eta \in \mbox{Val}_\cdom$ such that
    $\eta' \in \Solc{\pi'}$ holds for some $\eta' \in \mbox{Val}_\cdom$ verifying $\eta =_{\setminus X} \eta'$.
    \end{itemize}
$\pi$ is called {\em satisfiable} over $\cdom$ iff $\Solc{\pi} \neq \emptyset$,
and $\pi$ is called {\em unsatisfiable} over $\cdom$ iff $\Solc{\pi} = \emptyset$.
\item
{\em The solution set} $\Solc{\Pi}$ of a set $\Pi$ of constraints is defined as $\bigcap_{\pi \in \Pi} \Solc{\pi}$.
In this way, finite sets of constraints are interpreted as the conjunction of their members.
$\Pi$ is called {\em satisfiable} over $\cdom$ iff $\Solc{\Pi} \neq \emptyset$,
and $\Pi$ is called {\em unsatisfiable} over $\cdom$ iff $\Solc{\Pi} = \emptyset$.
\item
A constraint $\pi$ {\em is entailed} by a set of constraints $\Pi$ 
(in symbols,  $\Pi \model{\cdom} \pi$) iff $\Solc{\Pi} \subseteq \Solc{\pi}$.  \mathproofbox
\end{enumerate}
\end{defn}


The following example illustrates the previous definition:

\begin{exmp}[Constraint solutions and constraint entailment over $\rdom$]
\label{exmp:rconsol}
Consider the set of constraints 
$\Pi = \{cp_{\geq}(A,3.0),\, op_{+}(A,A,X), op_{\times}(2.0,A,Y)\} \subseteq \Con{\rdom}$. Then:
\begin{enumerate}
\item
For any valuation $\eta \in \mbox{Val}_\rdom$:
$\eta \in \Sol{\rdom}{\Pi}$ holds iff $\eta(A)$, $\eta(X)$ and $\eta(Y)$ are real numbers $a, x, y \in \REAL$ such that
$a \geq 3.0$, $a+a = x$ and $2.0\times a = y$.
\item
Due to the previous item, the following $\rdom$-entailments are valid:
  \begin{enumerate}
  \item
  $\Pi \model{\rdom} cp_{>}(X,5.5)$, because $\Sol{\rdom}{\Pi} \subseteq \Sol{\rdom}{cp_{>}(X,5.5)}$.
  \item
  $\Pi \model{\rdom} X == Y$, because $\Sol{\rdom}{\Pi} \subseteq \Sol{\rdom}{$X == Y$}$.
  \item
  $\Pi \model{\rdom} c(X) == c(Y)$, because $\Sol{\rdom}{\Pi} \subseteq \Sol{\rdom}{$c(X) == c(Y)$}$.
  Here we assume $c \in DC^1$.
  \item
  $\Pi \model{\rdom} [X,Y] == [Y,X]$, because $\Sol{\rdom}{\Pi} \subseteq \Sol{\rdom}{$[X,Y] == [Y,X]$}$.
  Here, the terms $[X,Y]$ and $[Y,X]$ are built from variables and list constructors.  \mathproofbox
  \end{enumerate} 
\end{enumerate}
\end{exmp}


The next technical result will be useful later on:

\begin{lem} [Substitution Lemma]
\label{lem:sl}
Assume a set of constraints $\Pi \subseteq \Con{\cdom}$ and a $\cdom$-substitution $\sigma$. Then:
\begin{enumerate}
\item
For any valuation $\eta \in \mbox{Val}_\cdom$: $\eta \in \Solc{\Pi\sigma} \iff \sigma\eta \in \Solc{\Pi}$.
\item
For any constraint $\pi \in \Con{\cdom}$: $\Pi \model{\cdom} \pi \Longrightarrow \Pi\sigma \model{\cdom} \pi\sigma$.
\end{enumerate}
\end{lem}

\begin{proof*}
Let us give a separate reasoning for each item.
\begin{enumerate}
\item
The following statement holds for any constraint $\pi  \in \Con{\cdom}$: 
$$(\star)\,\,\, \eta \in \Solc{\pi\sigma} \iff \sigma\eta \in \Solc{\pi}$$
In fact, $(\star)$ can can be easily proved reasoning by induction on the syntactic structure of $\pi$. 
Now, using $(\star)$ we can reason as follows:
$$
\begin{array}{c}
\eta \in \Solc{\Pi\sigma} \iff \eta \in \Solc{\pi\sigma} \mbox{ for all } \pi \in \Pi \iff_{\hspace{-2mm}(\star)} \\
\sigma\eta \in \Solc{\pi} \mbox{ for all } \pi \in \Pi \iff \sigma\eta \in \Solc{\Pi} \\
\end{array}
$$
\item
Assume $\Pi \model{\cdom} \pi$. 
For the sake of proving $\Pi\sigma \model{\cdom} \pi\sigma$, also assume an arbitrary $\eta \in \Solc{\Pi\sigma}$.
Then we get $\sigma\eta \in \Solc{\Pi}$ because of item 1 
and $\sigma\eta \in \Solc{\pi}$ due to the assumption $\Pi \model{\cdom} \pi$,
which implies $\eta \in \Solc{\pi\sigma}$ again because of item 1.
Since $\eta$ is arbitrary, we have proved $\Solc{\Pi\sigma} \subseteq \Solc{\pi\sigma}$,
i.e. $\Pi\sigma \model{\cdom} \pi\sigma$. \mathproofbox
\end{enumerate} 
\end{proof*} 

\subsubsection{Term equivalence w.r.t. a given constraint set}
\label{sec:cdoms:Pi-equiv}

Given two terms $t$, $s$ we will use the notation $t  \approx_{\Pi} s$ (read as $t$ and $s$ are $\Pi$-{\em equivalent})
as an abbreviation of $\Pi \model{\cdom} t == s$, assuming that the constraint domain $\cdom$ and the constraint set $\Pi \subseteq \Con{\cdom}$ are known. For the sake of simplicity, $\cdom$ is not made explicit in the $\approx_{\Pi}$ notation.
In this subsection we present some properties related to $\approx_{\Pi}$ which will be needed later.
First, we prove that $\approx_{\Pi}$ is an equivalence relation with a natural characterization.

\begin{lem} [$\Pi$-Equivalence Lemma]
\label{lem:Pi-equiv}
\begin{enumerate}
\item
$\approx_{\Pi}$ is an equivalence relation over $\TTerm$. 
\item
For any given terms $t$ and $s$ the following two statements are equivalent:
      \begin{enumerate}
      \item
      $t  \approx_{\Pi} s$.
      \item
     For any common position $p \in \mbox{pos}(t) \cap \mbox{pos}(s)$ some of the cases below holds:
	\begin{enumerate}
	\item
	$t \circ p$ or $s \circ p$ is a variable, and moreover $t|_p \approx_{\Pi} s|_p$.
	\item
	$t \circ p = s \circ p = u$ for some $u \in B_{\cdom}$.
	\item
	$t \circ p = s \circ p = c$ for some $n \in \mathbb{N}$ and some $c \in DC^n$.
	\end{enumerate}
     \end{enumerate}	
\item
$\approx_{\Pi}$ boils down to the syntactic equality relation $=$ when $\Pi$ is the empty set.
\end{enumerate}
\end{lem}
\begin{proof*}
We give a separate reasoning for each item.
\begin{enumerate}
\item
Checking that $\approx_{\Pi}$ satisfies the axioms of an equivalence relation
(i.e. {\em reflexivity}, {\em symmetry} and {\em transitivity}) is quite obvious.
\item
Due to Definition \ref{def:consol},
$t  \approx_{\Pi} s$ holds iff $t\eta$ and $s\eta$ are identical ground terms for each $\eta \in \Solc{\Pi}$.
This statement can be proved equivalent to condition 2.(b) reasoning by induction on $\Vert t \Vert + \Vert s \Vert$.
\item
Note that $t  \approx_{\emptyset} s$ holds iff $t\eta$ and $s\eta$ 
are identical ground terms for each $\eta \in \Solc{\emptyset} = \mbox{Val}_\cdom$.
This can happen iff $t$ and $s$ are syntactically identical. \mathproofbox
\end{enumerate}
\end{proof*} 

Since the set $\Var$ of all variables is countably infinite, we can assume an arbitrarily fixed bijective mapping $\mbox{ord} : \Var \to \mathbb{N}$. By convention, $\mbox{ord}(X)$ is called the ordinal number of $X$.
The notions defined below rely on this convention.

\begin{defn}[$\Pi$-Canonical Variables and Terms]
\label{def:cterms}
\begin{enumerate}
\item
A variable $X$ is called {\em $\Pi$-canonical} iff there is no other variable $X'$ such that $X \approx_{\Pi} X'$ and $\mbox{ord}(X') < \mbox{ord}(X)$.
\item
For each variable $X$ its {\em $\Pi$-canonical form} $\mbox{cf}_\Pi(X)$ is defined as the member of the set 
$\{X' \in \Var \mid X \approx_{\Pi} X'\}$ with the least ordinal number.
\item 
A term $t$ is called  {\em $\Pi$-canonical} iff all the variables occurring in $t$ are $\Pi$-canonical.
\item
For each term $t$ its {\em $\Pi$-canonical form} $\mbox{cf}_\Pi(t)$ is defined as the result of replacing $\mbox{cf}_\Pi(X)$
for each variable $X$ occurring in $t$. \mathproofbox
\end{enumerate}
\end{defn}

The following lemma states some obvious properties of terms in canonical form:

\begin{lem} [$\Pi$-Canonicity Lemma]
\label{lem:canon}
For each term $t$, $\mbox{cf}_\Pi(t)$ is $\Pi$-canonical and such that $t \approx{_\Pi} \mbox{ cf}_\Pi(t)$. 
Moreover, $t$ and $\mbox{cf}_\Pi(t)$ have the same positions and structure, except that each variable $X$ occurring 
at some position $p \in \mbox{vpos}(t)$ is replaced by an occurrence of $\mbox{cf}_\Pi(X)$ at the same position p in $\mbox{cf}_\Pi(t)$.
\end{lem}
\begin{proof}
Straightforward consequence of the construction of $\mbox{cf}_\Pi(t)$ from $t$ and the $\Pi$-Equiva\-lence Lemma \ref{lem:Pi-equiv}.
\end{proof} 

Given two terms $t$ and $s$, the term built from $t$ by replacing within $t$ each variable $X$ occurring at some position $p \in \mbox{vpos}(t) \cap \mbox{pos}(s)$ by the subterm $s|_p$ is called the {\em extension of $t$ w.r.t. to $s$} and noted as $t \ll s$ (or equivalently, $s \gg t$). A more precise definition of this notion and some related properties are given below.

\begin{defn}[Term extension]\label{defn:term-extension}
Given any two terms $t$ and $s$, the {\em extension of $t$ w.r.t. $s$} is defined by recursion on the syntactical structure of $t$:
\begin{itemize}
  \item 
   $X \ll s = s$ for each variable $X \in \Var$.
  \item
  $u \ll s  = u$ for each basic value $u \in B$.
  \item
 $c(t_1, \ldots, t_n) \ll s  = c(t_1 \ll s_1, \ldots, t_n \ll s_n)$ if $c \in DC^n$ 
 and there is some $c' \in DC^n$ such that $s = c'(s_1, \ldots, s_n)$.
 \item 
 $c(t_1, \ldots, t_n) \ll s = c(t_1, \ldots, t_n)$ if $c \in DC^n$ 
 and there is no $c' \in DC^n$ such that $s = c'(s_1, \ldots, s_n)$.
  \mathproofbox
\end{itemize}
\end{defn}

\begin{lem} [Extension Lemma]
\label{lem:ext}
The term extension operation $\ll$ enjoys the two following properties:
\begin{enumerate}
\item
{\em Symmetrical Extension Property:} \\
Let  $t'$, $t''$   be $\Pi$-canonical terms such that $t' \approx_{\Pi} t''$. 
Under this assumption $(t' \ll t'') = (t'' \ll t')$.
\item
{\em $\Pi$-Equivalence Extension Property:} \\
Let the terms $t$, $s$ be such that for any $p \in \mbox{pos}(t)$ with $t|_p = X \in \Var$ one has 
$p \in \mbox{pos}(s)$ and $X \approx_\Pi s|_p$.
Under this assumption $t \approx_{\Pi} (t \ll s)$.
\end{enumerate}
\end{lem}
\begin{proof*}[Proof of Symmetrical Extension Property]
Recall that the hypothesis $t' \approx_{\Pi} t''$ means that $\Pi \model{\cdom} t' == t''$.
We reason by complete induction on $\Vert t' \Vert + \Vert t'' \Vert$.
There are five possible cases:
	\begin{enumerate}
	\item $t' == t''$ is $c'(\ntup{t'}{n}) == c''(\ntup{t''}{n})$ for some $n \in \mathbb{N}$, $c', c'' \in DC^n$.
	In this case, the $\Pi$-Equivalence Lemma \ref{lem:Pi-equiv} ensures that $c' = c'' = c \in DC^n$ and 
	$t'_i \approx_{\Pi} t''_i$ holds for all $1 \leq i \leq n$.
	Clearly, the terms $t'_i$, $t''_i$ are $\Pi$-canonical. 
	Therefore, by induction hypothesis we can assume $(t'_i \ll t''_i) = (t''_i \ll t'_i)$ for all $1 \leq i \leq n$.
	Then, by definition of $\ll$ we get $t' \ll t'' = c(t'_1 \ll t''_1, \ldots, t'_n \ll t''_n) = c(t''_1 \ll t'_1, \ldots, t''_n \ll t'_n) = t'' \ll t'$ .
	\item $t' == t''$ is $u' == u''$ for some $u',  u'' \in B$.
	In this case, $u' \approx_{\Pi} u''$ implies that $u' = u'' = u \in B$,
	and by definition of $\ll$ we get $t' \ll t'' = t'' \ll t' = u \ll u = u$.
	\item $t' == t''$ is $X == Y$ for some $X, Y \in \Var$.
	In this case, $X \approx_{\Pi} Y$ and $X$, $Y$ $\Pi$-canonical implies that $X$, $Y$ must be identical variables.
	By definition of  $\ll$ we get $t' \ll t'' = t'' \ll t' = X \ll X = X$.
	\item $t' == t''$ is $X == t''$ with $X \in \Var$, $t'' \notin \Var$.
	In this case, by definition of $\ll$ we get $t' \ll t'' = X \ll t'' = t'' \mbox{ and } t'' \ll t' = t'' \ll X = t''$ .
	\item $t' == t''$ is $t' == Y$ with $Y \in \Var$, $t' \notin \Var$.
	In this case, by definition of $\ll$ we get $t' \ll t'' = t' \ll Y = t' = t'' \mbox{ and } t'' \ll t' = Y \ll t' = t'$ . \mathproofbox
	\end{enumerate}
\end{proof*}
\begin{proof*}[Proof of $\Pi$-Equivalence Extension Property]
Recall that the thesis $t \approx_{\Pi} (t \ll s)$ means that $\Pi \model{\cdom} t == (t \ll s)$.
We reason by complete induction on $\Vert t \Vert$.
There are four possible cases:
\begin{enumerate}
	\item $t$ is a variable $X \in \Var$.
	In this case, $X \ll s = s$ by definition of $\ll$, and $X \approx_{\Pi} s$ holds by hypothesis.
	\item $t$ is a basic value $u \in B$.
	In this case, $u \ll s = u$ by definition of $\ll$, and $u \approx_{\Pi} u$ holds trivially.
	\item $t$ is $c(\ntup{t}{n})$ for some $c \in DC^n$ and there is no $c' \in DC^n$ such that $s$ has the form $c'(\ntup{s}{n})$.
	In this case, $c(\ntup{t}{n}) \ll s = c(\ntup{t}{n})$ by definition of $\ll$, 
	and $c(\ntup{t}{n}) \approx_{\Pi}\, c(\ntup{t}{n})$ holds trivially.
	\item $t$ is $c(\ntup{t}{n})$ for some $c \in DC^n$ and $s$ is $c'(\ntup{s}{n})$ for some $c' \in DC^n$.
	In this case $c(\ntup{t}{n}) \ll c'(\ntup{s}{n}) = c(t_1 \ll s_1, \ldots, t_n \ll s_n)$ by definition of $\ll$.
	Moreover, the assumptions of the $\Pi$-Equivalent Extension Property hold for the smaller terms $t_i$, $s_i$ $(1 \leq i \leq n)$.
	By induction hypothesis we can assume $t_i \approx_{\Pi}\, (t_i \ll s_i)$ for all $1 \leq i \leq n$.
	Therefore, $c(\ntup{t}{n}) \approx_\Pi c(t_1 \ll s_1, \ldots, t_n \ll s_n)$ due to the $\Pi$-Equivalence Lemma \ref{lem:Pi-equiv}.
	\mathproofbox
	\end{enumerate}
\end{proof*} 

%% file: J2_2.tex
\subsection{Qualification Domains}
\label{sec:domains:qdoms}

The intended role of {\em Qualification Domains}  in an extended logic programming scheme SQCLP have been already explained in the Introduction.
They were originally introduced  in \cite{RR08} and their axiomatic definition was extended with axioms for an additional operation $\oslash$ in \cite{RR08b} in order to enable a particular implementation technique for program clauses with threshold conditions in their bodies.
The definition given below is again closer to the original one: $\oslash$ is omitted and the axioms of the operator $\circ$ are slightly refined.
 
\begin{defn}[Qualification Domains]
\label{def:qd}
A {\em Qualification Domain} is  any structure $\qdom = \langle D, \dleq, \bt, \tp, \circ \rangle$ verifying the following requirements:
\begin{enumerate}
  \item $D$, noted as $D_\qdom$ when convenient, is a set of elements called \emph{qualification values}.
    \item $\langle D, \dleq, \bt, \tp \rangle$ is a lattice with extreme points $\bt$ (called {\em infimum} or {\em bottom} element) and $\tp$ (called {\em maximum} or {\em top} element) w.r.t. the partial ordering $\dleq$, called {\em qualification ordering}. For given elements  $d, e \in D$, we  write $d \sqcap e$ for the {\em greatest lower bound} ($glb$) of $d$ and $e$, and $d \sqcup e$ for the {\em least upper bound} ($lub$) of $d$ and $e$. We also write $d \dlt e$ as abbreviation for $d \dleq e \land d \neq e$.
    \item $\circ : D \times D \rightarrow D$, called {\em attenuation operation}, verifies the following axioms:
        \begin{enumerate}
            \item $\circ$ is associative, commutative and monotonic w.r.t. $\dleq$.
            \item $\forall d \in D : d \circ \tp = d$.
            \item $\forall d \in D : d \circ \bt = \bt$.        
            \item $\forall d, e \in D  : d \circ e \dleq e$. 
            \item $\forall d, e_1, e_2 \in D : d \circ (e_1 \sqcap e_2) = (d \circ e_1) \sqcap (d \circ e_2)$. \mathproofbox
        \end{enumerate} 
\end{enumerate}
\end{defn}


Actually, some of the properties of $\circ$ postulated as axioms in the previous definition are redundant.\footnote{The authors are thankful to G. Gerla for pointing out this fact.} More precisely:
 
\begin{prop}[Redundant postulates of Qualification Domains]
\label{prop:qd-properties}
The properties (3)(c) and (3)(d) are redundant and can be derived from the other axioms in Definition \ref{def:qd}.
\end{prop}
\begin{proof}
Note that $\circ$ is commutative and monotonic w.r.t. $\dleq$ because of axiom (3)(a).
Since $\tp$ is the top element of the lattice, $d \dleq \tp$ holds for any $d \in D$. 
By monotonicity of $\circ$, $d \circ e \dleq \tp \circ e$ also holds for any $e \in D$.
By commutativity of $\circ$ and axiom (3)(b), 
$d \circ e \dleq \tp \circ e$ is the same as $d \circ e \dleq e$.
Therefore (3)(d) is a consequence of the other axioms postulated for $\circ$. 
In particular, taking $e = \bt$ we get $d \circ \bt \dleq \bt$,
which implies $d \circ \bt = \bt$ because $\bt$ is the bottom element of the lattice. 
Hence, (3)(c) also follows form the other axioms.
\end{proof}


In the rest of the report, $\qdom$ will generally denote an arbitrary qualification domain. 
For any finite $S = \{e_1, e_2, \ldots, e_n\} \subseteq D$, the {\em greatest lower bound} 
(also called {\em infimum} of $S$ and noted as $\infi S$) 
exists and can be computed as $e_1 \sqcap e_2 \sqcap \cdots \sqcap e_n$ (which reduces to $\top$ in the case $n = 0$). 
The dual claim concerning {\em least upper bounds} is also true.
As an easy consequence of the axioms, one gets the identity $d \circ \infi S =  \infi \{d \circ e \mid e \in S\}$.


Many useful qualification domains are such that $\forall d, e \in D \setminus \{\bt\} : d \circ e \neq \bt$.
In the sequel, any qualification domain $\qdom$ that verifies this property will be called {\em stable}. 
Below we present some basic qualification domains which are clearly stable, along with brief explanations of their role for building extended CLP languages as instances of the SQCLP scheme proposed in this report. 
Checking that these domains satisfy the axioms given in Def. \ref{def:qd} is left as an easy exercise.
In fact, the axioms have been chosen as a natural  generalization of some basic properties satisfied by the ordering $\leq$ and the operation $\times$ over the real interval $[0,1]$.

\subsubsection{The Domain $\B$ of Classical Boolean Values} \label{ssec:bd}

This domain is $\B \eqdef \langle \{0,1\}, \leq, 0, 1, \land \rangle$, 
where $0$ and $1$ stand for the two classical truth values {\em false} and {\em true}, 
$\leq$ is the usual numerical ordering over $\{0,1\}$, and $\land$ stands for the classical conjunction operation over $\{0,1\}$. 

\subsubsection{The Domain $\U$ of Uncertainty Values and its variant $\U'$} \label{ssec:ud}

This domain is $\U \eqdef \langle \mbox{U}, \leq, 0, 1,\times \rangle$, 
where $\mbox{U} = [0,1] = \{d \in \REAL \mid 0 \le d \le 1\}$, $\le$ is the usual numerical ordering,
and $\times$ is the multiplication operation. 
The top element $\tp$ is $1$ and the greatest lower bound $\infi S$ of a finite $S \subseteq \mbox{U}$ 
is the minimum value $\mbox{min}(S)$, which is $1$ if $S = \emptyset$.
Elements of $\U$ are intended to represent certainty degrees as used in \cite{VE86}.

A slightly different domain $\U'$ can be defined as  $\langle \mbox{U}, \leq, 0, 1, \mbox{min} \rangle$ where the only difference with respect to $\U$ is that in the case of $\U'$, $\circ = \mbox{min}$.

\subsubsection{The Domain $\W$ of Weight Values and related variants} \label{ssec:wd}

This domain is $\W \eqdef \langle \mbox{P}, \ge, \infty, 0, + \rangle$, 
where $\mbox{P} = [0,\infty] = \{d \in \REAL \cup \{\infty\} \mid d \ge 0\}$, 
$\geq$ is the reverse of the usual numerical ordering (with $\infty \ge d$ for any $d \in \mbox{P}$),
and $+$ is the addition operation (with $\infty + d = d + \infty = \infty$ for any $d \in \mbox{P}$). 
The top element $\tp$ is $0$ and the greatest lower bound $\infi S$ of a finite $S \subseteq \mbox{P}$ 
is the maximum value $\mbox{max}(S)$, which is $0$ if $S = \emptyset$.
Elements of $\W$ are intended to represent proof costs, measured as the weighted depth of proof trees.

In analogy to the definition of $\U'$ as a variant of $\U$, we can define a qualification domain $\W'$ as $\langle \mbox{P}, \ge, \infty, 0, \mbox{max} \rangle$ with $\circ = \mbox{max}$. 
Also, as a discrete variant of $\W$, we define the qualification domain $\W_d \eqdef \langle \mbox{P}, \ge, \infty, 0, + \rangle$ with the only difference w.r.t. $\W$ being that $\mbox{P} = \NAT \cup \{\infty\}$.
Elements of $\W_d$ are also intended to represent proof costs (represented by natural numbers in this case).
Finally, a variant $\W'_d$ of $\W_d$ can be defined by replacing the attenuation operation in $\W_d$ by max.

\subsubsection{Two product constructions} \label{ssec:pc}

To close this section, we present two product constructions that can be used to build compound qualification domains.
The mathematical definition is as follows:


\begin{defn}[Products of Qualification Domains]
\label{def:qdprod}
Let two qualification domains $\qdom_i = \langle D_i, \dleq_i, \bt_i, \tp_i, \circ_i \rangle$ ($i \in \{1, 2\}$) be given.
\begin{enumerate}
  \item  
  The {\em cartesian product} $\qdom_1\!\times\!\qdom_2$ 
  is defined as $\qdom \eqdef \langle D, \dleq, \bt, \tp, \circ \rangle$ where 
  $D \eqdef D_1 \!\times D_2$, 
  the partial ordering $\dleq$ is defined as $(d_1,d_2) \dleq (e_1,e_2) \diff$ $d_1 \dleq_1 e_1$ and $d_2 \dleq_2 e_2$, 
  $\bt \eqdef (\bt_1, \bt_2)$, $\tp \eqdef (\tp_1, \tp_2)$ 
  and the attenuation operator $\circ$ is defined as $(d_1,d_2) \circ (e_1,e_2) \eqdef (d_1 \circ_1 e_1, d_2 \circ_2 e_2)$.
  \item
  Given two elements $d_1 \in D_1$ and $d_2 \in D_2$, the {\em strict pair}  $\spair{d_1}{d_2}$ 
  is defined by case distinction as follows: 
  if $d_1 \neq \bt_1$ and $d_2 \neq \bt_2$, then $\spair{d_1}{d_2} = (d_1,d_2)$;
  if $d_1 = \bt_1$ or $d_2 = \bt_2$, then $\spair{d_1}{d_2} = (\bt_1,\bt_2)$.
  In both cases, $\spair{d_1}{d_2} \in D_1 \!\times D_2$.
  \item
  The {\em strict cartesian product} $\qdom_1 \!\otimes \qdom_2$ 
  is defined as $\qdom \eqdef \langle D, \dleq, \bt, \tp, \circ \rangle$ where
  $D  = D_1 \otimes D_2 \eqdef \{\spair{d_1}{d_2} \mid d_1 \in D_1, d_2 \in D_2\}$ 
  (or equivalently, $D = ((D_1\setminus \{\bt_1\}) \times (D_2 \setminus \{\bt_2\})) \cup \{(\bt_1, \bt_2)\}$), 
  the partial ordering $\dleq$ is defined as $(d_1,d_2) \dleq (e_1,e_2) \diff d_1 \dleq_1 e_1$ and $d_2 \dleq_2 e_2$, 
  $\bt \eqdef \spair{\bt_1}{\bt_2} = (\bt_1, \bt_2)$,
  $\tp \eqdef \spair{\tp_1}{\tp_2}$,
  and the attenuation operator $\circ$ is defined as $(d_1,d_2) \circ (e_1,e_2) \eqdef (d_1 \circ_1 e_1, d_2 \circ_2 e_2)$.
  Note the special case when $D_1$ or $D_2$  is  a singleton set.
  Then, $D$ is the singleton set $\{(\bt_1,\bt_2)\}$,
  $\spair{\tp_1}{\tp_2} = (\bt_1,\bt_2)$,
  and $(\tp_1, \tp_2) \in D$ happens to be false if one of the two sets $D_1$, $D_2$ is not a singleton. \mathproofbox
\end{enumerate}
\end{defn}

Intuitively, each value $(d_1,d_2)$ belonging to a product domain $\qdom_1 \!\times \qdom_2$ or $\qdom_1\otimes \qdom_2$ imposes the qualification $d_1$ {\em and also} the qualification $d_2$. In particular, values $(c,d)$ belonging to the product domains $\UW$ and $\U \otimes \W$ impose two qualifications, namely: a certainty value greater or equal than $c$ and a proof tree with weighted depth less or equal than $d$. This intuition indeed corresponds to the declarative  semantics formally defined in Section \ref{sec:sqclp}. 

The next theorem shows that the class of the qualification domains is closed under ordinary cartesian products, while the subclass of stable qualification domains is closed under strict cartesian products.
We are particularly interested in stable qualification domains built from basic domains by reiterated strict products,
because they can be encoded into into constraint domains in the sense explained in Subsection \ref{sec:cdoms:encoding} below.

\begin{thm}\label{thm:pd}
Assume two given qualification domains $\qdom_1$ and $\qdom_2$. 
Then the ordinary cartesian product $\qdom_1 \!\times \qdom_2$ is always a qualification domain.
Moreover, if $\qdom_1$ and $\qdom_2$ are stable, then the strict cartesian product $\qdom_1 \!\otimes \qdom_2$ is a stable qualification domain.
\end{thm}
\begin{proof*}
Here we reason only for the case of  the strict cartesian product
since the reasonings needed for the ordinary cartesian product are very similar and even simpler.
Assume that $\qdom_1$ and $\qdom_2$ are stable qualification domains, and let 
$\qdom = \qdom_1 \otimes \qdom_2$ be constructed as in Definition \ref{def:qdprod}.
In order to show that $\qdom$ is a stable qualification domain, we prove the four items below.
The assumption that $\qdom_1$ and $\qdom_2$ satisfy all the axioms from Definition \ref{def:qd}
is used in all the reasonings, often implicitly.

\begin{enumerate}
\item
The attenuation operator $\circ$ of $\qdom$ is well defined.
Assume $(d_1,d_2), (e_1,e_2) \in D$.  
According to Definition \ref{def:qdprod}, $(d_1,d_2) \circ (e_1,e_2)$ is defined as $(d_1 \circ_1 e_1, d_2 \circ_2 e_2)$. 
Since $D  = D_1 \otimes D_2$ is a strict subset of $D_1 \times D_2$, we must prove that $(d_1 \circ_1 e_1, d_2 \circ_2 e_2) \in D$.
We reason by distinction of cases:
    \begin{enumerate}
    \item[1.1.] $(d_1,d_2) = (\bt_1,\bt_2)$ or $(e_1,e_2) = (\bt_1,\bt_2)$.
    In this case, $(d_1 \circ_1 e_1, d_2 \circ_2 e_2) = (\bt_1,\bt_2) \in D$.
    \item[1.2.]  $(d_1,d_2) \neq (\bt_1,\bt_2)$ and $(e_1,e_2) \neq (\bt_1,\bt_2)$.
    In this case, $d_1, e_1 \in D_1 \setminus \{\bt_1\}$ and $d_2, e_2 \in D_2 \setminus \{\bt_2\}$.
    The assumption that $\qdom_1$ and $\qdom_2$ are stable ensures $d_1 \circ_1 e_1 \neq \bt_1$ and $d_2 \circ_2 e_2 \neq \bt_2$,
    and therefore $(d_1 \circ_1 e_1, d_2 \circ_2 e_2) \in D$.
    \end{enumerate}
\item 
$\langle D, \dleq, \bt, \tp \rangle$ is a lattice with extreme points 
$\bt = \spair{\bt_1}{\bt_2} = (\bt_1, \bt_2)$ and $\tp \eqdef \spair{\tp_1}{$ $\tp_2}$
w.r.t. the partial ordering $\dleq$.
By definition,  $(d_1,d_2) \dleq (e_1,e_2) \iff d_1 \dleq_1 e_1 \land d_2 \dleq_2 e_2$.
The fact that $\dleq$ is a partial ordering with minimum (bottom) element $\bt$ is an obvious connsequence.
To prove that $\tp$ is the maximum (top) element, we reason by case distinction. 
If $D_1$ is a singleton set, then $D_1 = \{\bt_1\}$, $\tp_1 = \bt_1$, $D = \{(\bt_1,\bt_2)\}$, 
and $\spair{\tp_1}{\tp_2} = (\bt_1,\bt_2)$ is obviously the top element.
The case that $D_2$ is a singleton set is argued similarly.
Finally, if neither $D_1$ nor $D_2$ are singleton, we have $\tp_1 \neq \bt_1$, $\tp_2 \neq \bt_2$,
and $\tp = \spair{\tp_1}{\tp_2} = (\tp_1,\tp_2)$ is clearly the top element.

To show that $\langle D, \dleq, \bt, \tp \rangle$ is a lattice, we assume two arbitrary elements $(d_1,d_2)$, $(e_1,e_2) \in D$,
and we prove:
       \begin{enumerate}
       \item[2.1.] There is a  $lub$ $(d_1,d_2) \sqcup (e_1,e_2) \in D$.
       The $lubs$ $d_1 \sqcup_1 e_1 \in D_1$ and $d_2 \sqcup_2 e_2 \in D_2$ are known to exist.
       We claim that $(d_1,d_2) \sqcup (e_1,e_2) = (d_1 \sqcup_1 e_1,d_2 \sqcup_2 e_2)$.
       Due to the component-wise definition of $\dleq$, it suffices to show that $(d_1 \sqcup_1 e_1,d_2 \sqcup_2 e_2) \in D$.
       We prove this by case distinction:
             \begin{enumerate}
             \item[2.1.1.] If $(d_1,d_2) = (\bt_1,\bt_2)$ then $(d_1 \sqcup_1 e_1,d_2 \sqcup_2 e_2) = (e_1,e_2) \in D$.
             \item[2.1.2.] If $(e_1,e_2) = (\bt_1,\bt_2)$ then $(d_1 \sqcup_1 e_1,d_2 \sqcup_2 e_2) = (d_1,d_2) \in D$.
             \item[2.1.3.] If $(d_1,d_2) \neq (\bt_1,\bt_2)$ and $(e_1,e_2) \neq (\bt_1,\bt_2)$ then the construction of $D$ ensures
             that $d_1, e_1 \in D_1 \setminus \{\bt_1\}$ and $d_2, e_2 \in D_2 \setminus \{\bt_2\}$. This implies
             $d_1 \sqcup_1 e_1 \neq \bt_1$ and $d_2 \sqcup_2 e_2 \neq \bt_2$, which guarantees $(d_1 \sqcup_1 e_1,d_2 \sqcup_2 e_2) 
             \in D$.
             \end{enumerate}
       \item[2.2.] There is a  $glb$ $(d_1,d_2) \sqcap (e_1,e_2) \in D$.
       The $glbs$ $d_1 \sqcap_1 e_1 \in D_1$ and $d_2 \sqcap_2 e_2 \in D_2$ are known to exist.
       We claim that $(d_1,d_2) \sqcap (e_1,e_2) = \spair{d_1 \sqcap_1 e_1}{d_2 \sqcap_2 e_2}$.
       We prove the claim by case distinction:
             \begin{enumerate}
              \item[2.2.1.] If $d_1 \sqcap_1 e_1 \neq \bt_1$ and $d_2 \sqcap_2 e_2 \neq \bt_2$, then 
              $\spair{d_1 \sqcap_1 e_1}{d_2 \sqcap_2 e_2}$ is the same as $(d_1 \sqcap_1 e_1,d_2 \sqcap e_2) \in D$,
              and this pair is the $glb$ of $(d_1,d_2)$ and $(e_1,e_2)$ due to the component-wise definition of $\dleq$.
              \item[2.2.2.] If $d_1 \sqcap_1 e_1 = \bt_1$ or $d_2 \sqcap_2 e_2 = \bt_2$, then
              $\spair{d_1 \sqcap_1 e_1}{d_2 \sqcap_2 e_2} = (\bt_1,\bt_2)$ is obviously a common lower bound of $(d_1,d_2)$ 
              and $(e_1,e_2)$.
              In order to conclude that $(\bt_1,\bt_2)$ is the $glb$ of $(d_1,d_2)$ and $(e_1,e_2)$,
              we show that $(\bt_1,\bt_2)$ is the only common lower bound of $(d_1,d_2)$ and $(e_1,e_2)$
              by the following reasoning: assume an arbitrary $(x,y) \in D$ such that
              $(x,y) \dleq  (d_1,d_2)$ and  $(x,y) \dleq  (e_1,e_2)$. Then $x \dleq d_1 \sqcap_1 e_1$ and $y \dleq  d_2 \sqcap_2 e_2$.
              Since $d_1 \sqcap_1 e_1 = \bt_1$ or $d_2 \sqcap_2 e_2 = \bt_2$, it follows that $x = \bt_1$ or  $y = \bt_2$. 
              By construction of $D$, it must be the case that $x = \bt_1$ {\em and}  $y = \bt_2$, 
              because otherwise $(x,y)$ would not belong to $D$. Therefore $(x,y) = (\bt_1,\bt_2)$, as desired.      
              \end{enumerate}
       \end{enumerate}
\item 
$\circ$ satisfies  axioms required for attenuation operators in Definition \ref{def:qd}.
By definition of $\circ$ we know
$$(\star) ~ (d_1,d_2) \circ (e_1,e_2) = (d_1 \circ_1 e_1, d_2 \circ_2 e_2)$$
which always belongs to $D$ as already proved in item $(1)$ above. 
All the axioms listed under item (3) of Definition \ref{def:qd} except (3)(e) follow easily from the equation $(\star)$
and the corresponding axioms for $\circ_1$ and $\circ_2$. In order to verify axiom (3)(e) for $\circ$, we 
assume three pairs $(d_1,d_2), (e_1,e_2), (e_1',e_2') \in D$. We must prove the equation
$$(\dag) ~ (d_1,d_2) \circ ((e_1,e_2) \sqcap (e_1',e_2')) = (d_1,d_2) \circ (e_1,e_2) \sqcap (d_1,d_2) \circ (e_1',e_2') \enspace .$$
We reason by case distinction:
      \begin{enumerate}
      \item[3.1.] If $d_1 = \bt_1$ and $d_2 = \bt_2$ then both sides of $(\dag)$ are equal to $(\bt_1,\bt_2)$,
      as shown by the following calculations:
      $$
      \begin{array}{l}
        (d_1,d_2) \circ ((e_1,e_2) \sqcap (e_1',e_2')) = (\bt_1,\bt_2) \circ ((e_1,e_2) \sqcap (e_1',e_2')) = (\bt_1,\bt_2) \\[2mm]
        (d_1,d_2) \circ (e_1,e_2) \sqcap (d_1,d_2) \circ (e_1',e_2') =  \\
        \qquad (\bt_1,\bt_2) \circ (e_1,e_2) \sqcap (\bt_1,\bt_2) \circ (e_1',e_2') = (\bt_1,\bt_2) \sqcap (\bt_1,\bt_2) = (\bt_1,\bt_2)
      \end{array}
      $$
      \item[3.2.] If the previous case does not apply, the construction of $D$ ensures that $d_1 \neq \bt_1$ and $d_2 \neq \bt_2$.
      We distiguish two subcases:
           \begin{enumerate}
           \item[3.2.1.] If $e_1 \sqcap_1 e_1' = \bt_1$ or $e_2 \sqcap_2 e_2' = \bt_2$ we get also
           $d_1 \circ_1 (e_1 \sqcap_1 e_1') = \bt_1$ or $d_2 \circ_2 (e_2 \sqcap_2 e_2') = \bt_2$,
           and we can assume the following:
           $$
           \begin{array}{ll}
             (\clubsuit) & \spair{e_1 \sqcap_1 e_1' }{e_2 \sqcap_2 e_2'} = (\bt_1,\bt_2) \\
             (\spadesuit) & \spair{d_1 \circ_1 (e_1 \sqcap_1 e_1')}{d_2 \circ_2 (e_2 \sqcap_2 e_2')} = (\bt_1,\bt_2) \\
           \end{array}
           $$
           We can now prove that both sides of $(\dag)$ are equal to $(\bt_1,\bt_2)$ as follows:
           $$
           \begin{array}{l}
           (d_1,d_2) \circ ((e_1,e_2) \sqcap (e_1',e_2')) = \\
           \qquad (d_1,d_2) \circ \spair{e_1 \sqcap_1 e_1' }{e_2 \sqcap_2 e_2'} =_{\clubsuit} (d_1,d_2) \circ (\bt_1,\bt_2) = (\bt_1,\bt_2) \\[2mm]
           (d_1,d_2) \circ (e_1,e_2) \sqcap (d_1,d_2) \circ (e_1',e_2') = \\
           \qquad (d_1 \circ_1 e_1, d_2 \circ_2 e_2) \sqcap (d_1 \circ_1 e_1', d_2 \circ_2 e_2') = \\
           \qquad \spair{d_1 \circ_1 e_1 \sqcap_1 d_1 \circ_1 e_1'}{d_2 \circ_2 e_2 \sqcap_2 d_2 \circ_2 e_2'} = \\
           \qquad \spair{d_1 \circ_1 (e_1 \sqcap_1 e_1')}{d_2 \circ_2 (e_2 \sqcap_2 e_2')} =_{\spadesuit} (\bt_1,\bt_2) \\
           \end{array}
           $$
           \item[3.2.2.] If $e_1 \sqcap_1 e_1' \neq \bt_1$ and $e_2 \sqcap_2 e_2' \neq \bt_2$ then the stability assumption made for $\qdom_1$ and $\qdom_2$
           ensures $d_1 \circ_1 (e_1 \sqcap_1 e_1') \neq \bt_1$ and $d_2 \circ_2 (e_2 \sqcap_2 e_2') \neq \bt_2$,
           and we can assume the following:
           $$
           \begin{array}{ll}
           (\diamondsuit) & \spair{e_1 \sqcap_1 e_1'}{e_2 \sqcap_2 e_2'} = (e_1 \sqcap_1 e_1',e_2 \sqcap_2 e_2') \\
           (\heartsuit) & \spair{d_1 \circ_1 (e_1 {\sqcap_1} e_1')}{d_2 \circ_2 (e_2 {\sqcap_2} e_2')} = (d_1 \circ_1 (e_1 {\sqcap_1} e_1'),d_2 \circ_2 (e_2 {\sqcap_2} e_2')) \\
           \end{array}
           $$
          Then,  $(\dag)$ is proved by the following calculations:
          $$
          \begin{array}{l}
          (d_1,d_2) \circ ((e_1,e_2) \sqcap (e_1',e_2')) = \\
          \qquad (d_1,d_2) \circ \spair{e_1 \sqcap_1 e_1'}{e_2 \sqcap_2 e_2'} =_{\diamondsuit} (d_1,d_2) \circ (e_1 \sqcap_1 e_1', e_2 \sqcap_2 e_2') = \\
          \qquad (d_1 \circ_1 (e_1 \sqcap_1 e_1'), d_2 \circ_2 (e_2 \sqcap_2 e_2')) \\[2mm]
          (d_1,d_2) \circ (e_1,e_2) \sqcap (d_1,d_2) \circ (e_1',e_2') = \\
          \qquad (d_1 \circ_1 e_1, d_2 \circ_2 e_2) \sqcap (d_1 \circ_1 e_1', d_2 \circ_2 e_2') = \\
          \qquad \spair{d_1 \circ_1 e_1 \sqcap_1 d_1 \circ_1 e_1'}{d_2 \circ_2 e_2 \sqcap_2 d_2 \circ_2 e_2'} = \\
          \qquad \spair{d_1 \circ_1 (e_1 \sqcap_1 e_1')}{d_2 \circ_2 (e_2 \sqcap_2 e_2')} =_\heartsuit \\
          \qquad (d_1 \circ_1 (e_1 \sqcap_1 e_1'), d_2 \circ_2 (e_2 \sqcap_2 e_2')) 
          \end{array}
          $$
           \end{enumerate}
      \end{enumerate}     
\item
$\qdom_1 \!\otimes \qdom_2$ is stable. To prove this let us assume  
$(d_1,d_2), (e_1,e_2) \in D_1 \otimes D_2 \setminus \{(\bt_1,\bt_2)\}$. 
Then $d_1,e_1 \in D_1 \setminus \{\bt_1\}$ and $d_2, e_2 \in D_2 \setminus \{\bt_2\}$.
Since $\qdom_1$ and $\qdom_2$ are stable qualification domains, we can infer that
$d_1 \circ_1 e_1 \neq \bt_1$ and $d_2 \circ_2 e_2 \neq \bt_2$,
which implies 
$(d_1,d_2) \circ (e_1,e_2) = (d_1 \circ_1 e_1, d_2 \circ_2 e_2) \neq (\bt_1,\bt_2).$ \mathproofbox
\end{enumerate}
\end{proof*}

 
\subsubsection{Encoding Qualification Domains into Constraint Domains}
\label{sec:cdoms:encoding}

In this subsection we investigate a technical relationship between qualification domains and constraint  domains which will play a key role in the rest of the report.

\begin{defn} [Expressing $\qdom$ in $\cdom$]
\label{dfn:expressible}
A qualification domain $\qdom$ with carrier set $D_\qdom$ is expressible in a constraint domain $\cdom$ with carrier set $C_\cdom$  if 
there is an injective mapping $\imath : \aqdomd{\qdom} \to C_\cdom$ embedding $\aqdomd{\qdom}$ into  $C_\cdom$,
and the two following requirements are satisfied:
  \begin{enumerate}
     \item 
     There is a $\cdom$-constraint $\qval{X}$ such that $\Solc{\qval{X}}$ is the set of all
     $\eta \in \mbox{Val}_\cdom$ such that  $\eta(X)$ belongs to the range of $\imath$.
    \item 
    There is a $\cdom$-constraint $\qbound{X,Y,Z}$ encoding ``$x \dleq y \circ z$'' in the following sense:
    any $\eta \in \mbox{Val}_\cdom$ satisfying 
    $\eta(X) = \imath(x)$, $\eta(Y) = \imath(y)$  and $\eta(Z) = \imath(z)$ for some $x, y, z \in \aqdom$
    verifies $\eta \in \Solc{\qbound{X,Y,Z}} \Longleftrightarrow x \dleq y \circ z$.
  \end{enumerate}
Moreover, if $\qval{X}$ and $\qbound{X,Y,Z}$ can be chosen as existential constraints, we say that $\qdom$ is 
{\em existentially expressable} in $\cdom$. \mathproofbox
\end{defn}

In the sequel, $\cdom$-constraints built as instances of $\qval{X}$ and $\qbound{X,Y,Z}$ are called {\em qualification constraints}, and $\Omega$ is used as notation for sets of qualification constraints.
The following result ensures that several qualification domains built with the techniques presented in Subsection \ref{sec:domains:qdoms} are existentially expressible in $\mathcal{H}$, $\rdom$ or $\mathcal{FD}$, according to the case.

\begin{prop}[Expressing Qualification Domains in Constraint Domains]
\label{prop:inrdom}
\begin{enumerate}
\item
The domain $\B$ is existentially expressible in any given constraint domain $\cdom$.
\item
The domains $\U$, $\U'$\!, $\W$ and $\W'$\! are existentially expressible in $\rdom$
(or any other constraint domain that supports the expressivity of $\rdom$).
\item
The domains $\W_d$ and $\W'_d$ are existentially expressible in $\fdom$
(or any other constraint domain that supports the expressivity of $\fdom$).
%
\item
Assume that the two qualification domains $\qdom_1$ and $\qdom_2$ are stable and (existentially) expressible in a constraint domain $\cdom$.
Then, $\qdom_1\! \otimes \qdom_2$ is also (existentially) expressible in $\cdom$. 
\end{enumerate}
\end{prop}

\begin{proof*}
\begin{enumerate}
\item
Straightforward, due to the fact  that $D_\B \setminus \{\bt\}$ is the singleton set $\{\tp\} = \{true\}$.
\item 
We prove that $\U$ can be existentially expressed in $\rdom$ as follows: 
$D_\U \setminus \{\bt\} = D_\U \setminus \{0\} = (0,1] \subseteq \REAL \subseteq C_\rdom$;
therefore $\imath$ can be taken as the identity embedding mapping from $(0,1]$ into $\REAL$.
Moreover,  $\qval{X}$ can be built as the existential $\rdom$-constraint $cp_<(0,X) \land cp_\leq(X,1)$
and $\qbound{X,Y,Z}$ can be built as the existential $\rdom$-constraint $\exists X_1(op_\times(Y,Z,X_1) \land cp_\leq(X,X_1))$.
By very similar reasonings it is easy to check that $\U'$, $\W$ and $\W'$ can also be existentially expressed in $\rdom$.
Note that in the cases of $\W$ and $\W'$ there is no reasonable way to define $\imath(\infty)$.
This is the reason why the domain of $\imath$ is required to be $\aqdom$ in Definition \ref{dfn:expressible}.
\item
Note that $D_{\W_d} \setminus \{\bt\} = D_{\W_d} \setminus \{\infty\} = \NAT$.
Moreover, $\dleq$ is $\geq$  and $\circ$ is $+$ in $\W_d$.
Therefore, $\W_d$ can be expressed in $\fdom$ by taking $\imath$ as the identity embedding mapping, building  $\qval{X}$ as an existential $\fdom$ constraint that requires the value of $X$ to be an integer $x \geq 0$,
and building $\qbound{X,Y,Z}$ as an existential $\fdom$ constraint that requires the values of $X$, $Y$ and $Z$ to be integers $x$, $y$ and $z$ such that $x \geq y + z$. 
A similar reasoning proves that $\W'_d$ is existentially expressible in $\fdom$ also.
\item
For $j = 1,2$ assume the existence of injective embedding mappings $\imath_j$ and $\cdom$-constraints 
$\qvali{j}{X}$, $\qboundi{j}{X,Y,Z}$ that can be used to (existentially) express $\qdom_j$ in $\cdom$.
Due to Theorem \ref{thm:pd} we know that $\qdom_1\! \otimes \qdom_2$ is a stable qualification domain.
Moreover, because of the construction of $\qdom = \qdom_1 \otimes \qdom_2$ given in 
Definition \ref{def:qdprod}, we know that $\aqdom = (D_1\setminus \{\bt_1\}) \times 
(D_2 \setminus \{\bt_2\})$. We also know that $\dleq$ is defined component-wise from $\dleq_1$, 
and $\dleq_2$, and analogously for $\circ$. Therefore, $\qdom$ can be (existentially) expressed in $\cdom$ by taking:
  \begin{itemize}
  \item
  $\imath$ defined by $\imath(d_1,d_2) \eqdef (\imath_1(d_1),\imath_2(d_2))$.
  \item
  $\qval{X}$ built as the prenex form of the constraint 
  $$\exists X_1 \exists X_2 (X == (X_1,X_2) \land \qvali{1}{X_1} \land \qvali{2}{X_2})$$
  which is existential if $\qvali{1}{X_1}$ and $\qvali{2}{X_2}$ are both existential.
  \item
  $\qbound{X,Y,Z}$ built as the prenex form of the constraint
  $$
  \begin{array}{c}
  \exists X_1 \exists X_2 \exists Y_1 \exists Y_2 \exists Z_1 \exists Z_2 
  (X \!== (X_1,X_2) \land Y \!== (Y_1,Y_2) \land Z \!== (Z_1,Z_2) \land \\
  \qboundi{1}{X_1,Y_1,Z_1} \land \qboundi{2}{X_2,Y_2,Z_2})
  \end{array}$$
  which is existential if $\qboundi{1}{X_1,Y_1,Z_1}$ and $\qboundi{2}{X_2,Y_2,$ $Z_2}$ are both existential.
  \end{itemize}  
Note that this reasoning does not work for the non-strict cartesian product $\qdom = \qdom_1 \times \qdom_2$,
because in this case $\aqdom = (D_1 \times D_2) \setminus \{(\bt_1,\bt_2)\}$ includes some pairs $(d_1,d_2)$ 
such that either $d_1 = \bt_1$ or $d_2 = \bt_2$ (but not both), and the given mappings
$\imath_1$, $\imath_2$ cannot be used to embed such pairs into $C_\cdom$. \mathproofbox
\end{enumerate}
\end{proof*}


%% file: J2_3.tex
\subsection{Similarity and Proximity Relations}
\label{sec:domains:simrels}


{\em Similarity relations} over a given set $S$ have been defined in \cite{Zad71,Ses02} and related literature as mappings $\simrel : S \times S \to [0,1]$ that satisfy reflexivity, symmetry and transitivity  axioms analogous to those required for classical equivalence relations. A more general notion called {\em proximity relation} has been defined in \cite{DP80} by omitting the transitivity axiom. Each value $\simrel(x,y)$ computed by a similarity (resp. proximity) relation $\simrel$ is called the {\em similarity degree} (resp. {\em proximity degree}) between $x$ and $y$. In our previous paper \cite{CRR08},  we proposed to generalize similarity relations by allowing  elements of an arbitrary qualification domain $\qdom$ to serve as proximity degrees. The definition below further generalizes this approach by considering proximity relations.

\begin{defn}[Proximity and similarity relations]
\label{defn:simrel}
Let a qualification domain $\qdom$ with carrier set $D$ and a set $S$ be given.
\begin{enumerate}
\item
A {\em $\qdom$-valued relation} over $S$ is any mapping $\simrel : S \times S \to D$.
\item
A  $\qdom$-valued relation $\simrel$ over $S$ is called    
   \begin{enumerate}
        \item {\em Reflexive}  iff $\forall x \in S : \simrel(x,x) = \tp$.
        \item {\em Symmetrical}  iff $\forall x,y \in S: \simrel(x,y) = \simrel(y,x)$.
        \item \emph{Transitive} iff $\forall x,y,z  \in S: \simrel(x,z) \dgeq \simrel(x,y) \sqcap \simrel(y,z)$.
   \end{enumerate}
\item
$\simrel$ is called a {\em $\qdom$-valued proximity relation} iff $\simrel$ is reflexive and symmetrical.
\item
If $\simrel$ is also transitive, then it is called a {\em $\qdom$-valued similarity relation}. 
\item
$\simrel$ is called {\em finitary} iff there are only finitely many choices of elements 
$x, y \in S$ such that $x \neq y$ and $\simrel(x,y) \neq \bt$.
From a practical viewpoint, this is a very natural requirement.  \enspace  \mathproofbox
\end{enumerate}    
\end{defn}


Obviously, $\qdom$-valued similarity relations are a particular case of $\qdom$-valued proximity relations.
Moreover, when $\qdom$ is chosen as  the qualification domain $\U$,  the previous definition provides proximity and similarity relations in the sense of \cite{Zad71,DP80}. In this case, a proximity degree $\simrel(x,y) = d \in [0,1]$ can be naturally interpreted as a {\em certainty degree} for the assertion that $x$ and $y$ are interchangeable. 
On the other hand,  if $\simrel$ is $\W$-valued, then $\simrel(x,y) = d \in [0,\infty]$ can be interpreted as a {\em cost} to be paid for $y$ to play the role of $x$.
More generally, the proximity degrees computed by a $\qdom$-valued proximity relation may have different interpretations according to the intended role of $\qdom$-elements as qualification values.


In contrast to previous works such as \cite{Ses02,CRR08}, in the rest of this report we will work with $\qdom$-valued proximity rather than similarity relations. Formally, this leads to  more general results. Moreover, as already noted by  \cite{SM99} and other authors, the transitivity property required for similarity relations may be counterintuitive in some cases. For instance, assume nullary constructors \texttt{colt}, \texttt{cold} and \texttt{gold} intended to represent words composed of four letters. Then, measuring the proximity between such words might reasonably lead to a $\U$-valued proximity relation $\simrel$ such that  $\simrel(\texttt{colt},\texttt{cold}) = 0.9$, $\simrel(\texttt{cold},\texttt{gold}) = 0.9$ and $\simrel(\texttt{colt},\texttt{gold}) = 0.4$.
On the other hand, insisting on  $\simrel$ to  be transitive would enforce the unreasonable condition $\simrel(\texttt{colt},\texttt{gold}) \geq 0.9$. 
Therefore, a similarity relation would be not appropriate in this case.


The special mapping $\sid : S \times S \to D$ defined as $\sid(x,x) = \tp$ for all $x \in S$ and $\sid(x,y) = \bt$ for all $x,y \in S$, $x \neq y$ is trivially a $\qdom$-valued similarity (and therefore, also proximity) relation called the \emph{identity}. 

\subsubsection{Admissible triples and proximity relations}
\label{sec:domains:proximity}

From now on, we will focus on proximity relations that are related to constraint domains in the following  sense:

\begin{defn}[Admissible triples]
\label{defn:simrel:admissible}
$\langle \simrel,\qdom,\cdom \rangle$ is called an {\em admissible triple} iff the following requirements are fulfilled:
  \begin{enumerate}
  \item
  $\cdom$ is a constraint domain  with signature $\Sigma = \langle DC, DP, PP \rangle$ and set of  basic values $B_\cdom$,
  and $\qdom$ is a qualification domain expressible in $\cdom$ in the sense of Definition \ref{dfn:expressible}.
  \item 
  $\simrel$ is a $\qdom$-valued proximity relation over  $S = \Var \uplus B_\cdom \uplus DC \uplus DP \uplus PP$.
   \item 
   $\simrel$ restricted to $\Var$ behaves as the identity, 
   i.e. $\simrel(X,X) = \tp$ for all $X \in \Var$ and $\simrel(X,Y) = \bt$ for all $X, Y \in \Var$ such that $X \neq Y$.
    \item 
    For any  $x, y \in S$, $\simrel(x,y) \neq \bt$ can happen only if some of the following cases holds:
      \begin{enumerate}
        \item $x = y$ are identical.
        \item $x,y \in B_\cdom$ are basic values.
        \item $x,y \in DC$ are data constructor symbols with the same arity.
        \item $x,y \in DP$ are defined predicate symbols with the same arity.
      \end{enumerate}
      In particular, $\simrel(p,p') \neq \bt$ cannot happen if $p, p'  \in PP$ are syntactically different primitive predicate symbols. 
      \enspace \mathproofbox
  \end{enumerate}
\end{defn}


In the rest of the report, our notions and results are valid for any choice of an admissible triple  $\langle \simrel,\qdom,\cdom \rangle$.  Proposition \ref{prop:inrdom} provides useful information for building 
admissible triples. For any given admissible triple, $\simrel$ can be naturally extended to act over terms and atoms over $\cdom$. The extension, also noted $\simrel$, works as specified in the recursive definition below. An analogous definition  for the case of $\U$-valued similarity relations can be found in \cite{Ses02}.  

\begin{defn}[$\simrel$ acting over terms and atoms]
\label{def:ER}
For any given admissible triple,  $\simrel$ is extended to work over $\cdom$-terms and $\cdom$-atoms as follows:
\begin{enumerate}
\item For any $t \in \TTerm$: \\
$\simrel(t,t) = \tp$.
\item For $X \in \Var$ and for any term $t$ different from $X$: \\
$\simrel(X,t) = \simrel(t,X) = \bt$.
\item For $c, c' \in DC$ with different arities $n$, $m$: \\
$\simrel(c(\ntup{t}{n}), c'(\ntup{t'}{m})) = \bt$.
\item For $c, c' \in DC$ with the same arity $n$:\\
$\simrel(c(\ntup{t}{n}), c'(\ntup{t'}{n})) = \simrel(c,c') \sqcap \simrel(t_1,t_1') \sqcap \ldots \sqcap \simrel(t_n,t_n')$.
\item For $r, r' \in DP \cup PP$ with different arities $n$, $m$: \\
$\simrel(r(\ntup{t}{n}), r'(\ntup{t'}{m})) = \bt$.
\item For $r, r' \in DP \cup PP$ with the same arity $n$:\\
$\simrel(r(\ntup{t}{n}), r'(\ntup{t'}{n}) = \simrel(r,r') \sqcap \simrel(t_1,t_1') \sqcap \ldots \sqcap \simrel(t_n,t_n')$. \mathproofbox
\end{enumerate}
\end{defn}

Given two terms $t, t'$ and some fixed qualification value $\lambda \in \aqdom$ 
we will use the notation $t  \approx_{\lambda} t'$ (read as $t$ and $t'$ are {\em $\simrel$-close at level $\lambda$})
as an abbreviation of $\lambda \dleq \simrel(t,t')$.
For the sake of simplicity, $\simrel$ is not made explicit in the $\approx_{\lambda}$ notation.
The following lemma provides a natural characterization of $\approx_{\lambda}$.
A similar result was given in \cite{Ses02} for the case of case of $\U$-valued similarity relations.

\begin{lem} [Proximity Lemma]
\label{lem:prox}
\begin{enumerate}
\item
$\approx_{\lambda}$ is a reflexive and symmetric equivalence relation over terms,
which is also transitive (and hence an equivalence relation) in the case that $\simrel$ is a similarity relation. 
\item
For any given terms $t$ and $t'$ the following two statements are equivalent:
      \begin{enumerate}
      \item
      $t \approx_{\lambda} t'$.
      \item
     $\mbox{pos}(t) = \mbox{pos}(t')$, and for each $p \in \mbox{pos}(t) \cap \mbox{pos}(t')$ some of the cases below holds:
	\begin{enumerate}
	\item
	$t \circ p = t' \circ p = X$ for some $X \in \Var$.
	\item
	$t \circ p = s \circ p = u$ for some $u \in B_{\cdom}$.
	\item
	$t \circ p = c$ and $t' \circ p = c'$ for some $n \in \mathbb{N}$ and some $c, c'  \in DC^n$ such that $\lambda \dleq \simrel(c,c')$.
	\end{enumerate}
     \end{enumerate}	
\item
Any given terms $t$ and $t'$ such that $t \approx_{\lambda} t'$ are {\em quasi-identical} in the following sense: 
$\mbox{pos}(t) = \mbox{pos}(t')$, and for each $p \in \mbox{pos}(t) = \mbox{pos}(t')$ either $t \circ p = t' \circ p$ or else 
$t \circ p$ and $t' \circ p$ are two data constructors of the same arity.
\item
$\approx_{\lambda}$ boils down to the syntactic equality relation `=' when $\simrel$ is the identity proximity relation $\sid$.
\end{enumerate}
\end{lem}
\begin{proof*}
We give a separate reasoning for each item.
\begin{enumerate}
\item
Note that the reflexivity and symmetry of $\approx_{\lambda}$ are a trivial consequence
of the reflexivity and symmetry of $\simrel$, as formulated in Definition \ref{defn:simrel}.
In the case that $\simrel$ is a similarity relation, transitivity of $\approx_{\lambda}$ follows from transitivity  of
$\simrel$ and the obvious fact that $\lambda \sqcap \lambda = \lambda$. 
\item
The claimed equivalence between conditions 2(a) and 2(b) can be proved reasoning by induction on $\Vert t \Vert + \Vert t' \Vert$.
\item
This item is an obvious consequence of the previous one.  
\item
Assume $\simrel = \sid$. 
Then, as a trivial consequence of Definition \ref{def:ER}, the value of $\simrel(t,t')$ is $\tp$ if $t = t'$ and $\bt$ otherwise. 
Since $\lambda \neq \bt$, it follows that $t \approx_{\lambda} t'$ iff  $t = t'$, as desired. \mathproofbox
\end{enumerate}
\end{proof*} 

The following result shows that $\approx_{\lambda}$ is compatible with the term extension operation in a natural way:

\begin{lem} [Proximity Preservation Lemma]
\label{lem:prox-preserve}
Assume terms $t$, $t'$ and $\lambda \in \aqdom$ such that $t \approx_{\lambda} t'$.
Then $(t \ll s)  \approx_{\lambda} (t' \ll s)$ holds also for any term $s$.
\end{lem}
\begin{proof}
Due to the assumption, $t \approx_{\lambda} t'$ are quasi-identical and satisfy condition 2(b)
as stated in the Proximity Lemma \ref{lem:prox}.
From this fact and Definition \ref{defn:term-extension}  it is quite clear that the same condition
2(b) holds also for $t \ll s$, $t' \ll s$ and $\lambda$.
Therefore, the Proximity Lemma allows to conclude $t \approx_{\lambda} t'$ as desired.
\end{proof} 

\subsubsection{Term proximity w.r.t. a given constraint set}
\label{sec:domains:Pproximity}

Reasoning with equations between $\cdom$-terms will require to infer information both from $\simrel$ and for some fixed  constraint set $\Pi \subseteq \Con{\cdom}$. 
This leads to a generalization of $\approx_{\lambda}$  formally defined as follows:

\begin{defn}[Term proximity w.r.t. a given constraint set]
\label{defn:Pi-prox}
Let $\langle \simrel,\qdom,\cdom \rangle$ be any admissible triple.
Assume  $\lambda \in \aqdom$ and $\Pi \subseteq \Con{\cdom}$.
We will say that $t$ and $s$ are {\em $\simrel$-close at level $\lambda$ w.r.t. $\Pi$}
(in symbols, $t  \approx_{\lambda, \Pi} s$)
iff  there are two terms $\hat{t}$, $\hat{s}$ such that 
$t \approx_{\Pi} \hat{t}$, $s \approx_{\Pi} \hat{s}$ and $\hat{t} \approx_{\lambda} \hat{s}$.
For the sake of simplicity neither $\simrel$ nor $\cdom$ are made explicit in the notation. 
\mathproofbox
\end{defn}


As illustration, let us present an  example using the constraint domain $\rdom$
and the qualification domain $\U$:
 
\begin{exmp}[Term proximity w.r.t.  $\rdom$ constraints]
\label{exmp:eqp}
Consider  $\Pi = \{op_{+}(A,A,X), op_{\times}(2.0,A,Y), Z == c(X,Y)\} \subseteq \Con{\rdom}$.
Note that this choice of $\Pi$ ensures $X \approx_{\Pi} Y$.
Assume $c, c' , c'' \in DC^2$ and an $\U$-valued proximity relation $\simrel$ 
such that $\simrel(c',c) = \simrel (c,c'') = 0.8$ and $\simrel(c',c'') =  0.6$. 
Then:
\begin{enumerate}
\item
$c(Y,X) \approx_{\Pi} Z$ holds, but $c'(Y,X) \approx_{\Pi} Z$ is false.
\item
$c'(Y,X) \approx_{0.7, \Pi} Z$ holds, because  
$c'(Y,X) \approx_{\Pi} c'(X,X)$, $Z \approx_{\Pi} c(X,X)$ and $c'(X,X) \approx_{0.7} c(X,X)$.
\item
$Z \approx_{0.7, \Pi} c''(X,Y)$  is also true, for similar reasons.
\item
$c'(Y,X) \approx_{0.7, \Pi} c''(X,Y)$ is false, 
because there is no possible choice of terms $\hat{t}$ and $\hat{s}$ such that 
$c'(Y,X) \approx_{\Pi} \hat{t}$, $c''(X,Y) \approx_{\Pi} \hat{s}$ and $\hat{t} \approx_{0.7} \hat{s}$. \enspace  \mathproofbox
\end{enumerate}
\end{exmp}


The next result states some basic properties of relations $\approx_{\lambda, \Pi}$.

\begin{lem} [$\Pi$-Proximity Lemma]
\label{lem:Pi-prox}
\begin{enumerate}
\item
$\approx_{\lambda, \Pi}$ is invariant w.r.t. $\approx_\Pi$ in the following sense:
$t \approx_{\lambda, \Pi} s$ implies $t' \approx_{\lambda, \Pi} s'$ for all terms
$t'\!,s'$ such that $t' \approx_\Pi t$ and $s' \approx_\Pi s$.
\item
$\approx_{\lambda, \Pi}$ is a reflexive and symmetric  relation over terms,
which is also transitive (and hence an equivalence relation) in the case that $\simrel$ is a similarity relation. 
\item
For any given terms $t$ and $t'$ the following two statements are equivalent:
  \begin{enumerate}
  \item
  $t \approx_{\lambda, \Pi} t'$.
  \item
  For any common position $p \in \mbox{pos}(t) \cap \mbox{pos}(t')$ some of the cases below holds:
    \begin{enumerate}
    \item
    $t \circ p$ or $t' \circ p$ is a variable, and moreover $t|_p \approx_{\lambda, \Pi} t'|_p$.
    \item
    $t \circ p = s \circ p = u$ for some $u \in B_{\cdom}$.
    \item
    $t \circ p = c$ and $t' \circ p = c'$ for some $n \in \mathbb{N}$ and some $c, c'  \in DC^n$ such that $\lambda \dleq \simrel(c,c')$.
    \end{enumerate}
  \end{enumerate}	
\item
$\approx_{\lambda, \Pi}$ boils down to $\approx_{\lambda}$ when $\Pi$ is the empty set, and
$\approx_{\lambda, \Pi}$ boils down to $\approx_{\Pi}$ when $\simrel$ is the identity proximity relation $\sid$.
\end{enumerate}
\end{lem}

\begin{proof*}
We give a separate reasoning for each item. 
Definition \ref{defn:Pi-prox} and Lemmata \ref{lem:Pi-equiv} and \ref{lem:prox} are implicitely used at some points.
\begin{enumerate}
\item
By definition,  $t \approx_{\lambda, \Pi} s$ means the existence of terms $\hat{t},\hat{s}$ such that 
$t \approx_{\Pi} \hat{t}$, $s \approx_{\Pi} \hat{s}$ and $\hat{t} \approx_{\lambda} \hat{s}$.
In case that $t' \approx_\Pi t$ and $s' \approx_\Pi s$, the same terms $\hat{t},\hat{s}$ 
verify $t' \approx_{\Pi} \hat{t}$, $s' \approx_{\Pi} \hat{s}$ (since $\approx_\Pi$ is an equivalence relation)
and $\hat{t} \approx_{\lambda} \hat{s}$. Therefore $t' \approx_{\lambda, \Pi} s'$.
\item 
Let us consider the three properties in turn: \\
{\em Reflexivity:} $t \approx_{\lambda, \Pi} t$ holds because $\hat{t} = t$  trivially verifies 
$t \approx_{\Pi} \hat{t}$ and $\hat{t} \approx_{\lambda} \hat{t}$. \\
{\em Symmetry:}  Assume $t \approx_{\lambda, \Pi} s$. 
Then there are terms  $\hat{t},\hat{s}$ such that $t \approx_{\Pi} \hat{t}$, $s \approx_{\Pi} \hat{s}$
and $\hat{t} \approx_{\lambda} \hat{s}$. Due to the symmetry of $\approx_{\lambda}$ we get 
$\hat{s} \approx_{\lambda} \hat{t}$ and hence $s \approx_{\lambda, \Pi} t$. \\
{\em Transitivity:} Example \ref{exmp:eqp} above shows that $\approx_{\lambda, \Pi}$ is not transitive in general.
Here we prove transitivity of  $\approx_{\lambda, \Pi}$ under the assumption that $\simrel$ is a similarity relation fulfilling the transitive property stated in Definition \ref{defn:simrel}. \\
Assume terms $t_1$, $t_2$ and $t_3$ such that $t_1 \approx_{\lambda, \Pi} t_2$ and $t_2 \approx_{\lambda, \Pi} t_3$.
Then there are terms $t'_1$, $t'_2$, $t''_2$ and $t''_3$ such that
$$(a) ~ t_1 \approx_{\Pi}  t'_1, ~ t_2 \approx_{\Pi} t'_2, ~ t'_1 \approx_{\lambda} t'_2
\quad \mbox{and} \quad
(b) ~ t_2 \approx_{\Pi}  t''_2, ~ t_3 \approx_{\Pi} t''_3, ~ t''_2 \approx_{\lambda} t''_3 \enspace .$$
Without loss of generality,  $t'_1$, $t'_2$, $t''_2$ and $t''_3$ can be assumed to be $\Pi$-canonical terms.
If they were not, it would suffice to to replace each of them by its $\Pi$-canonical form, built as explained in Definition \ref{def:cterms}.  
This replacement would preserve properties $(a)$ and $(b)$ thanks to the  Canonicity Lemma \ref{lem:canon}.

We claim that there are three terms $\hat{t_1}$, $\hat{t_2}$, and $\hat{t_3}$ such that
$$(c) ~ t_1 \approx_{\Pi} \hat{t_1}, ~ t_2 \approx_{\Pi} \hat{t_2}, ~  t_3 \approx_{\Pi}\hat{t_3}
\quad \mbox{and} \quad
(d) ~ \hat{t_1} \approx_{\lambda} \hat{t_2}, ~ \hat{t_2} \approx_{\lambda} \hat{t_3} \enspace .$$
Conditions $(c)$ and $(d)$ imply $t_1 \approx_{\lambda, \Pi} t_3$ due to Definition \ref{defn:Pi-prox} 
and the transivity property of $\approx_{\lambda}$, which is ensured by the transitivity of $\simrel$ and the Proximity Lemma \ref{lem:prox}.
In the rest of  this item we prove the claim by assuming $(a)$ and $(b)$ 
and showing how to build $\hat{t_1}$, $\hat{t_2}$, and $\hat{t_3}$ fulfilling $(c)$ and $(d)$.

Note that the assumption $t'_1 \approx_{\lambda} t'_2$ implies that $t'_1$ and $t'_2$ are quasi-identical terms, 
due Proximity Lemma \ref{lem:prox}(3).
Analogously, terms $t''_2$ and $t''_3$ must be also quasi-identical due to the assumption $t''_2 \approx_{\lambda} t''_3$,
and the target condition $(d)$ requires that $\hat{t_1}$, $\hat{t_2}$, $\hat{t_3}$ are constructed as quasi-identical terms.
Since our assumptions do not guarantee  quasi-identity of terms $t'_2$ and $t''_2$,
we resort to the term extension construction from Definition \ref{defn:term-extension} for building the terms $\hat{t_i}$.
More precisely, we build:
$$
\hat{t_1} =_{\mbox{def}} (t'_1 \ll t''_2); \quad 
\hat{t_2} =_{\mbox{def}} (t'_2  \ll t''_2) = (t''_2 \ll t'_2); \quad \mbox{and }
\hat{t_3} =_{\mbox{def}} (t''_3  \ll t'_2)
$$
where the identity $(t'_2  \ll t''_2) = (t''_2 \ll t'_2)$ is a consequence of the Symmetrical Extension Property 
from Lemma \ref{lem:ext}, which can be applied because $t'_2$ and $t''_2$ are $\Pi$-canonical 
and assumptions $(a)$, $(b)$ imply $t'_2 \sim_\Pi t''_2$.
We argue that conditions $(c)$ and $(d)$ are satisfied as follows:

--- Condition $(c)$, $t_1 \approx_{\Pi}  \hat{t_1}$:
By assumptions $(a)$, $(b)$  we know $t_1 \approx_{\Pi} t'_1$ and $t'_2 \approx_{\Pi} t''_2$. 
It suffices to prove $t'_1 \approx_{\Pi}  \hat{t_1}$.
For each $p \in \mbox{pos}(t'_1)$ with $t'_1|_p = X \in \Var$ we have $t'_2|_p = X$  because $t'_1$ and $t'_2$ are quasi-identical.
Moreover, $t'_2 \approx_{\Pi} t''_2$ implies that $p \in \mbox{pos}(t''_2)$ and $X \sim_\Pi t''_2|_p$,
due to the $\Pi$-Equivalence Lemma \ref{lem:Pi-equiv}. 
In these conditions, $t'_1 \approx_{\Pi}  \hat{t_1}$ follows from $\hat{t_1} = (t'_1 \ll t''_2)$
and the Equivalent Extension Property from Lemma  \ref{lem:ext}.

--- Condition $(c)$, $t_3 \approx_{\Pi}  \hat{t_3}$:
The proof for this is analogous to the previous one.
Since $t_3 \approx_{\Pi} t''_3$ and $\hat{t_3} = (t''_3  \ll t'_2)$ it suffices to prove $t''_3 \approx_{\Pi} (t''_3  \ll t'_2)$,
which can be done with the help of the Equivalent Extension Property.

--- Condition $(c)$, $t_2 \approx_{\Pi}  \hat{t_2}$:
By assumptions $(a)$, $(b)$ we know $t_2 \approx_{\Pi} t'_2$ and $t'_2 \approx_{\Pi} t''_2$. 
It suffices to prove $t'_2 \approx_{\Pi}  \hat{t_2}$.
For each $p \in \mbox{pos}(t'_2)$ with $t'_2|_p = X \in \Var$ we have $p \in \mbox{pos}(t''_2)$ and $X \sim_\Pi t''_2|_p$,
due to $t'_2 \approx_{\Pi} t''_2$ and the $\Pi$-Equivalence Lemma \ref{lem:Pi-equiv}. 
In these conditions, $t'_2 \approx_{\Pi}  \hat{t_2}$ follows from $\hat{t_2} = (t'_2 \ll t''_2)$
and the Equivalent Extension Property.

--- Condition $(d)$, $\hat{t_1} \approx_{\lambda} \hat{t_2}$:
By assumption $(a)$ we have $t'_1 \approx_{\lambda} t'_2$.
By the Proximity Preservation Lemma \ref{lem:prox-preserve} this implies $(t'_1 \ll t''_2)  \approx_{\lambda} (t'_2 \ll t''_2)$.
By construction of the terms $\hat{t_i}$,  this is the same as $\hat{t_1}  \approx_{\lambda} \hat{t_2}$.

--- Condition $(d)$, $\hat{t_2} \approx_{\lambda} \hat{t_3}$:
The proof for this is analogous to the previous one.
Assumption $(b)$ provides $t''_2 \approx_{\lambda} t''_3$.
Then, the Proximity Preservation Lemma guarantees $(t''_2 \ll t'_2)  \approx_{\lambda} (t''_3 \ll t'_2)$,
which is the same as $\hat{t_2}  \approx_{\lambda} \hat{t_3}$ by construction of the terms $\hat{t_i}$
(this time viewing $\hat{t_2}$ as $(t''_2 \ll t'_2)$ rather than $(t'_2  \ll t''_2)$ as in the previous argumentation).
\item
The claimed equivalence between conditions 3(a) and 3(b) can proved reasoning by induction on $\Vert t \Vert + \Vert t' \Vert$.
\item
According to Definition  \ref{defn:Pi-prox},  
$t \approx_{\lambda, \Pi} s$ is true  iff  $(\star)$ holds, where:
$$(\star) ~ \mbox{there are terms } \hat{t}, \hat{s} \mbox{ such that }  
t \approx_{\Pi} \hat{t}, ~ s \approx_{\Pi} \hat{s} \mbox{ and } \hat{t} \approx_{\lambda} \hat{s}.$$
Let us argue for the two cases $\Pi = \emptyset$ and $\simrel = \sid$ separately:
    \begin{itemize}
    \item
    Assume that $\Pi = \emptyset$.  
   Then,  due to $\Pi$-Equivalence Lemma \ref{lem:Pi-equiv}(3), 
   $(\star)$ can be rewritten as
   $$\mbox{there are terms } \hat{t}, \hat{s} \mbox{ such that }  
   t = \hat{t}, ~ s = \hat{s} \mbox{ and } \hat{t} \approx_{\lambda} \hat{s}$$
   which is equivalent  to $t \approx_{\lambda} s$.
    \item
   Assume now that $\simrel = \sid$. 
   Then, due to Proximity Lemma  \ref{lem:prox}(4),
   $(\star)$ can be rewritten as
   $$\mbox{there are terms } \hat{t}, \hat{s} \mbox{ such that }  
   t \approx_{\Pi} \hat{t}, ~ s \approx_{\Pi} \hat{s} \mbox{ and } \hat{t} =  \hat{s}$$
   which is equivalent to $t \approx_{\Pi} s$. \mathproofbox
    \end{itemize}
\end{enumerate}
\end{proof*} 


The following technical lemma will be needed later on. 
Although it is closely related to Lemma \ref{lem:sl}(2), it needs a separate proof because 
statements of the form $t \approx_{\lambda, \Pi} s$ are not $\cdom$-constraints.

\begin{lem} [Substitution Lemma for $\approx_{\lambda, \Pi}$]
\label{lem:slpeq}
Let $\langle \simrel,\qdom,\cdom \rangle$ be any admissible triple.
Assume $\lambda \in \aqdom$, $\Pi \subseteq \Con{\cdom}$, and two terms $t,s$ such that $t \approx_{\lambda, \Pi} s$.
Then $t\sigma \approx_{\lambda, \Pi\sigma} s\sigma$ holds for every $\cdom$-substitution $\sigma$.
\end{lem}

\begin{proof}
Because of the assumptions and Definition \ref{defn:Pi-prox}, there are terms $\hat{t},\hat{s}$ such that 
$t \approx_{\Pi} \hat{t}$ (i.e. $\Pi \model{\cdom} t == \hat{t}$), 
$s \approx_{\Pi} \hat{s}$ (i.e. $\Pi \model{\cdom} s == \hat{s}$)
and $\hat{t} \approx_{\lambda} \hat{s}$.
Consider now any substitution $\sigma$. Due to Lemma \ref{lem:sl}(2), we get
$\Pi\sigma \model{\cdom} t\sigma == \hat{t}\sigma$ 
(i.e. $t\sigma \approx_{\Pi\sigma} \hat{t}\sigma$)
and $\Pi\sigma \model{\cdom} s\sigma == \hat{s}\sigma$
(i.e. $s\sigma \approx_{\Pi\sigma} \hat{s}\sigma$). 
Moreover, $\hat{t}\sigma \approx_{\lambda} \hat{s}\sigma$
is an easy consequence of $\hat{t} \approx_{\lambda} \hat{s}$ and Proximity Lemma \ref{lem:prox}(2).
Then, Definition \ref{defn:Pi-prox} allows to conclude  $t\sigma \approx_{\lambda, \Pi\sigma} s\sigma$
simply by taking $\hat{t}\sigma$ as $\hat{t\sigma}$ and  $\hat{s}\sigma$ as $\hat{s\sigma}$. \mathproofbox
\end{proof}

%% file: J3_0.tex
\section{The SQCLP Programming Scheme}
\label{sec:sqclp}

In this section we develop the SQCLP scheme with instances $\sqclp{\simrel}{\qdom}{\cdom}$ 
announced in the introduction. 
The parameters $\simrel$, $\qdom$ and $\cdom$ stand for an admissible proximity relation, 
a qualification domain and a constraint domain with a certain signature $\Sigma$, respectively. 
By convention, we consider only those instances of the scheme whose parameters 
are chosen to constitute an {\em admissible triple} in the sense of Definition \ref{defn:simrel:admissible}.
We focus on declarative semantics, using  an interpretation transformer and a logical inference system
to provide alternative characterizations of least program models.
We also discuss declarative semantics of goals  and related approaches.  

A brief remark regarding notation is in place here. For the sake of notational consistency with previous works (either by us or other authors) where similarity rather than proximity relations were used, we keep the symbol $\simrel$ for proximity relations and the uppercase letter  \texttt{S} in the names of programming schemes. Our results, however, do not rely on the transitivity property from Definition \ref{defn:simrel}.

%% file: J3_1.tex
\subsection{Programs and their Declarative Semantics}
\label{sec:sqclp:programs}


A $\sqclp{\simrel}{\qdom}{\cdom}$-program is a set $\Prog$ of  \emph{qualified program rules} (also called \emph{qualified clauses}) of the form $C : A \qgets{\alpha} \qat{B_1}{w_1}, \ldots, \qat{B_m}{w_m}$, where $A$ is a defined atom, $\alpha \in \aqdomd{\qdom}$ is called the {\em attenuation factor} of the clause and each $\qat{B_j}{w_j} ~ (1 \le j \le m)$ is an atom $B_j$ annotated with a so-called {\em threshold value} $w_j \in \bqdomd{\qdom}$.
The intended meaning of $C$ is as follows:  
if for all $1 \leq j \leq m$ one has $\qat{B_j}{e_j}$ (meaning that $B_j$ holds with qualification value $e_j$)
for some $e_j \dgeq^? w_j$,
then $\qat{A}{d}$ (meaning that $A$ holds with qualification value $d$)
can be inferred for any $d \in \aqdom$ such that $d \dleq \alpha \circ \infi_{j = 1}^m e_j$. 
By convention, $e_j \dgeq^? w_j$ means $e_j \dgeq w_j$ if $w_j ~{\neq}~?$ and is identically true otherwise.
In practice threshold values equal to `?' and attenuation values equal to $\tp$ can be omitted.


As motivating example, consider a $\sqclp{\simrel}{\U\!\!\otimes\!\!\W}{\rdom}$-program $\Prog$ 
including the clauses and equations for $\simrel$ displayed in Figure \ref{fig:sample}. 
From Subsection \ref{sec:domains:qdoms} recall that qualification values in $\U\!\otimes\!\W$ are pairs
$(d,e)$ (where $d$ represents a certainty degree and $e$ represents a proof cost),
as well as the behavior of $\dleq$ and $\circ$ in $\U\!\otimes\!\W$. 
Consider the problem of proving $\qat{\texttt{goodWork(king\_liar)}}{(d,e)}$ from $\Prog$.
This can be achieved for $d = 0.75 \times \mbox{min}\{d_1,d_2\}$, $e = 3 + \mbox{max}\{e_1,e_2\}$
by using $R_1$ instantiated by $\{{\tt X} \mapsto {\tt king\_liar}, {\tt Y} \mapsto {\tt shakespeare}\}$,
and going on to prove  $\qat{\texttt{famousAuthor(shakespeare)}}{(d_1,e_1)}$
for some  $d_1 \geq 0.5$,  $e_1 \leq 100$ 
and $\qat{\texttt{wrote(shakespeare,king\_liar)}}{(d_2,e_2)}$ for some $d_2$, $e_2$.
Thanks to $R_2$, $R_3$ and $\simrel$, these proofs succeed with $(d_1,e_1) = (0.9,1)$ and $(d_2,e_2) = (0.8,2)$.
Therefore, the  desired proof succeeds with certainty degree  $d = 0.75 \times \mbox{min}\{0.9,0.8\} = 0.6$, 
and proof cost $e = 3 + \mbox{max}\{1,2\} = 5$.

\begin{figure}[h]
\figrule
\small
\flushleft
\hspace{3mm} $R_1$ : \verb+goodWork(X) <-(0.75,3)- famousAuthor(Y)#(0.5,100), wrote(Y,X)#?+\\
\hspace{3mm} $R_2$ : \verb+famousAuthor(shakespeare) <-(0.9,1)-+\\
\hspace{3mm} $R_3$ : \verb+wrote(shakespeare,king_lear) <-(1,1)-+\\[2mm]
\hspace{1cm} $\simrel$\verb+(king_lear,king_liar) = (0.8,2)+
\caption{$\sqclp{\simrel}{\,\U\!\otimes\!\W}{\rdom}$ Program Fragment}
\label{fig:sample}
\figrule
\vspace{-3mm}
\end{figure}


It is useful to define some special types of program clauses and programs, as follows:

\begin{itemize}
\item
A clause is called {\em attenuation-free} iff $\alpha = \tp$. The name is justified because $\tp$
is an identity element for the attenuation operator $\circ$, as explained in Subsection
\ref{sec:domains:qdoms}. By convention, attenuation-free clauses may be written with the simplified
syntax $A \gets \qat{B_1}{w_1}, \ldots, \qat{B_m}{w_m}$.
\item
A clause is called {\em threshold-free} iff $w_j =\,\, ?$ for all $j = 1 \ldots m$.
The name is justified because the threshold value $w_j =\,\, ?$ occurring as annotation of a body atom $B_j$
does not impose any particular requirement to the qualification value of $B_j$.
Threshold-free clauses may be written with the simplified syntax $A \qgets{\alpha}  B_1, \ldots, B_m$.
\item
A clause is called {\em qualification-free} iff it is both attenuation-free and threshold-free.
These clauses may be written with the simplified syntax $A \!\gets\! B_1, \ldots, B_m$.
They behave just like those used in the classical CLP scheme.
\item
A clause is called {\em constraint-free} iff all its body atoms are defined.
\item
A program is called  attenuation-free iff all its clauses are of this type.
Thresh\-old-free, qualification-free and constraint-free programs are defined similarly.
\end{itemize}


The more technical $\sqclp{\simrel}{\U}{\rdom}$-program $\Prog$ presented below will serve as a
\emph{running example} to illustrate various points in the rest of the report.

\begin{exmp}[Running example]
\label{exmp:running}
Assume unary constructors $c, c' \in DC^1$, binary predicate symbols $p, p', q \in DP^2$ and a ternary predicate symbol $r \in DP^3$.
Consider the admissible triple $\langle \simrel,\U,\rdom \rangle$,
where $\simrel$ is  an $\U$-valued proximity relation such that
$\simrel(c,c') = 0.9$ and $\simrel(p,p') = 0.8$.
Let $\Prog$ be the $\sqclp{\simrel}{\U}{\rdom}$-program consisting of the qualified clauses $R_1$, $R_2$ and $R_3$ listed below:
\begin{itemize}
  \item[] \small $R_1 : q(X,c(X)) \qgets{1.0}$
  \item[] \small $R_2 : p(c(X),Y) \qgets{0.9} \qat{q(X,Y)}{0.8}$
  \item[] \small $R_3 : r(c(X),Y,Z) \qgets{0.9} \qat{q(X,Y)}{0.8}, \qat{cp_{\geq}(X,0.0)}{?}$ \mathproofbox
 \end{itemize}
\end{exmp}


As we will see in the Conclusions, the classical $\mbox{CLP}$ scheme for Constraint Logic Programming originally introduced in \cite{JL87} can be seen as a particular case of the $\mbox{SQCLP}$ scheme.
In the rest of this subsection we present a declarative semantics for $\sqclp{\simrel}{\qdom}{\cdom}$-programs inspired by \cite{GL91,GDL95}. These papers provided three different program semantics $\mathcal{S}_{i}$ ($i = 1, 2, 3$)
characterizing {\em valid ground solutions for goals}, {\em valid open solutions for goals} and {\em computed answers for goals} in $\mbox{CLP}$, respectively.
In fact, the $\mathcal{S}_{i}$ semantics in \cite{GL91,GDL95} were conceived as the $\mbox{CLP}$ counterpart
of previously known semantics for logic programming,
namely the least ground Herbrand model semantics \cite{Apt90,Llo87},
the open Herbrand model semantics, also known as $\mathcal{C}$-semantics \cite{Cla79,FLMP93},
and the $\mathcal{S}$-semantics \cite{FLMP89,BGLM94};
see \cite{AG94} for a very  concise and readable overview.

In this report we restrict ourselves to develop a $\mathcal{S}_{2}$-like semantics which can be used
to characterize valid open solutions for SQCLP goals as we will see in Subsection \ref{sec:sqclp:goals}.
As a basis for our semantics we use so-called {\em qc-atoms} of the form $\cqat{A}{d}{\Pi}$,
intended to assert that the atom $A$ is entailed by the constraint set $\Pi$ with qualification degree $d$.
We also use a special entailment relation $\entail{\qdom,\cdom}$ intended to capture some implications between qc-atoms
whose validity depends neither on the proximity relation $\simrel$ nor on the semantics of defined predicates.
A formal definition of these notions is as follows:


\begin{defn}[qc-atoms, observables and $(\qdom, \cdom)$-entailment]
\label{defn:atoms-entail}
\begin{enumerate}
  \item \label{defn:atoms-entail:atoms}
  \emph{Qualified constrained atoms} (or simply \emph{qc-atoms}) are statements of the form $\cqat{A}{d}{\Pi}$,
  where $A \in \At$ is an atom, $d \in D$ is a qualification value, and $\Pi \subseteq \Con{\cdom}$ is a finite set of  constraints.
  \item \label{defn:atoms-entail:kinds}
  A qc-atom $\cqat{A}{d}{\Pi}$ is called {\em defined}, {\em primitive} or {\em equational} according to the syntactic form of $A$.
  \item \label{defn:atoms-entail:observable}
  A qc-atom $\cqat{A}{d}{\Pi}$ is called {\em observable} iff $d \in \aqdom$ and $\Pi$ is satisfiable.
  \item \label{defn:atoms-entail:entail}
  Given two qc-atoms  $\varphi : \cqat{A}{d}{\Pi}$ and $\varphi' : \cqat{A'}{d'}{\Pi'}$,
  we say that $\varphi$ $(\qdom, \cdom)$-\emph{entails} $\varphi'$ (in symbols, $\varphi \entail{\qdom,\cdom} \varphi'$)
  iff there is some $\cdom$-substitution $\theta$ satisfying $A' = A\theta$, $d' \dleq d$ and $\Pi' \model{\cdom} \Pi\theta$.
  \mathproofbox
\end{enumerate}
\end{defn}


We will focus our attention on observable qc-atoms
because they can be interpreted as observations of valid open solutions for atomic goals in $\sqclp{\simrel}{\qdom}{\cdom}$
as we will see in Subsection \ref{sec:sqclp:goals}.
The example below illustrates the main technical ideas from Definition \ref{defn:atoms-entail}.

\begin{exmp}[Observable qc-atoms and $(\qdom, \cdom)$-entailment]
\label{exmp:qc-atoms}
Consider the admissible triple underlying Example \ref{exmp:running} and the sets of $\rdom$-constraints:
$$
\begin{array}{c@{\hspace{1mm}}c@{\hspace{1mm}}l}
\Pi & = & \{cp_{>}(X,1.0),\, op_{+}(A,A,X),\, op_{\times}(2.0,A,Y)\} \\
\Pi' & = & \{cp_{\geq}(A,3.0),\, op_{\times}(2.0,A,X),\, op_{+}(A,A,Y)\}\\
\end{array}
$$
Then, the following are observable qc-atoms:
$$
\begin{array}{l@{\hspace{1cm}}l}
\varphi_1 = \cqat{q(X,c'(Y))}{0.9}{\Pi} & \varphi_3 = \cqat{r(c'(Y),c(X),Z)}{0.8}{\Pi} \\
\varphi_2 = \cqat{p'(c'(Y),c(X))}{0.8}{\Pi} & \varphi'_3 = \cqat{r(c'(Y),c(X),c(Z'))}{0.7}{\Pi'} \\
\end{array}
$$
and the $(\U, \rdom)$-entailment $\varphi_3 \entail{\U, \rdom} \varphi'_3$ is valid thanks to $\theta = \{Z \mapsto c(Z')\}$,
which satisfies $r(c'(Y),c(X),c(Z')) = r(c'(Y),c(X),Z)\theta$, $0.7 \leq 0.8$ and $\Pi' \model{\rdom} \Pi\theta$. \mathproofbox
\end{exmp}


The intended meaning of $\entail{\qdom,\cdom}$ as an entailment relation not depending on
the meanings of  defined predicates motivates the first item in the next definition.

\begin{defn}[Interpretations]
\label{defn:interpretations}
Let $\langle \simrel,\qdom,\cdom \rangle$ be any given admissible triple. Then:
\begin{enumerate}
\item
A {\em qualified constrained interpretation} (or {\em qc-interpretation}) is a set $\I$ of
observable defined qc-atoms closed under $(\qdom, \cdom)$-entailment.
In other words, a set $\I$ of qc-atoms which satisfies the following two conditions:
  \begin{enumerate}
  \item
  Each $\varphi \in \I$ is an observable defined qc-atom.
  \item
  If $\varphi \in \I$ and $\varphi'$ is another defined observable qc-atom such that $\varphi \entail{\qdom,\cdom} \varphi'$,
  then also $\varphi' \in \I$.
  \end{enumerate}
\item
Assume any given qc-interpretation $\I$.
For any observable qc-atom $\varphi$, we say that  $\varphi$ is valid in $\I$ modulo $\simrel$
(in symbols, $\I \isqchlrdc \varphi$) iff some of the three cases below holds:
\begin{enumerate}
  \item $\varphi$ is defined and $\varphi \in \I$.
  \item $\varphi : \cqat{(t == s)}{d}{\Pi}$ is equational and $t \approx_{d, \Pi} s$.
  \item $\varphi : \cqat{\kappa}{d}{\Pi}$ is primitive and $\Pi \model{\cdom} \kappa$.
  \enspace  \mathproofbox
\end{enumerate}
\end{enumerate}
\end{defn}

Note that a given interpretation $\I$ can include several observables $\cqat{A}{d_i}{\Pi}$ for the same (possibly not  ground) atom $A$ but is not required to include on ``optimal'' observable $\cqat{A}{d}{\Pi}$ with $d$ computed as the {\em lub} of all $d_i$.
By contrast, the other related works discussed in the Introduction view program interpretations as mappings $\I$ from the ground Herbrand base into some set of lattice elements (the real interval $[0,1]$ in many cases). In such interpretations, each ground atom $A$ has attached one single lattice element $d = \I(A)$ intended as ``the optimal qualification'' for $A$.
Our view of interpretations is closer to the expected operational behavior of goal solving systems and can be used to characterize the validity of  solutions computed by such systems, as we will see in Subsection \ref{sec:sqclp:goals}.

Note also that the notation $\I \isqchlrdc \varphi$ is defined only for the case that $\varphi$ is observable.
In the sequel, we will implicitly assume that $\varphi$ is observable in any context where the notation $\I \isqchlrdc \varphi$ is used. The next technical result shows that validity in any given interpretation is closed under entailment.

\begin{prop}[Entailment Property for Interpretations]
\label{prop:ep-i}
Assume that $\I \isqchlrdc \varphi$ and $\varphi \entail{\qdom,\cdom} \varphi'$. Then $\I \isqchlrdc \varphi'$.
\end{prop}
\begin{proof*}
Due to the hypothesis $\varphi \entail{\qdom,\cdom} \varphi'$
we can assume $\varphi = (\cqat{A}{d}{\Pi})$, $\varphi' = (\cqat{A'}{d'}{\Pi'})$
and some $\cdom$-substitution $\theta$ such that  $A' = A\theta$, $d' \dleq d$ and $\Pi' \model{\cdom} \Pi\theta$.
We now distinguish cases according to the syntactic form of $\varphi$:
\begin{enumerate}
\item $\varphi$ is defined.
In this case, $\varphi'$ is also defined.
Moreover, $\I \isqchlrdc \varphi$ is equivalent to $\varphi \in \I$ because of Definition \ref{defn:interpretations},
which implies $\varphi' \in \I$ because qc-interpretations are closed under $\entail{\qdom,\cdom}$,
which is equivalent to $\I \isqchlrdc \varphi'$  because of Definition \ref{defn:interpretations}.
\item $\varphi$ is equational.
In this case $A$ and $A'$ have the form $t == s$ and $t\theta == s\theta$, respectively.
Moreover, $\I \isqchlrdc \varphi$ is equivalent to
$t \approx_{d, \Pi} s$ because of Definition \ref{defn:interpretations},
which implies $t\theta \approx_{d, \Pi\theta} s\theta$ because of Lemma \ref{lem:slpeq},
which trivially  implies $t\theta \approx_{d', \Pi'} s\theta$
because of $\Pi' \model{\cdom} \Pi\theta$ and $d' \dleq d$,
which is equivalent to $\I \isqchlrdc \varphi'$  because of Definition \ref{defn:interpretations}.
\item $\varphi$ is primitive.
In this case $A$ and $A'$ have the form $\kappa$ and $\kappa\theta$, respectively.
Moreover, $\I \isqchlrdc \varphi$ is equivalent to $\Pi \model{\cdom} \kappa$ because of Definition \ref{defn:interpretations},
which implies $\Pi\theta \model{\cdom} \kappa\theta$ because of Lemma \ref{lem:sl},
which implies $\Pi' \model{\cdom} \kappa\theta$ because of $\Pi' \model{\cdom} \Pi\theta$,
which is equivalent to $\I \isqchlrdc \varphi'$  because of Definition \ref{defn:interpretations}.
\enspace \mathproofbox
\end{enumerate}
\end{proof*}


The definition below explains when a given  interpretation is regarded as a model of a given program,
as well as the related notion of semantic consequence.

\begin{defn}[Models and semantic consequence]
\label{defn:models}
Let a $\sqclp{\simrel}{\qdom}{\cdom}$-program $\Prog$ and an observable qc-atom 
$\varphi : \cqat{p'(\ntup{t'}{n})}{d}{\Pi}$ be given.
$\varphi$ is an {\em immediate consequence} of a qc-interpretation $\I$ via a program rule
$(R_l :  p(\ntup{t}{n}) \qgets{\alpha} \qat{B_1}{w_1}, \ldots, \qat{B_m}{w_m}) \in \Prog$ iff
there exist a $\cdom$-substitution $\theta$ and a choice of qualification values 
$d_0, d_1, \ldots, d_n, e_1, \ldots, e_m \in \aqdom$ such that:
\begin{enumerate}
\item[(a)]
$\simrel(p',p) = d_0$
\item[(b)]
$\I \isqchlrdc \cqat{(t'_i == t_i\theta)}{d_i}{\Pi}$ (i.e. $t'_i \approx_{d_i,\Pi} t_i\theta$) for $i = 1 \ldots n$
\item[(c)]
$\I \isqchlrdc \cqat{B_j\theta}{e_j}{\Pi}$ with $e_j \dgeq^? w_j$ for $j = 1 \ldots m$
\item[(d)]
$d \dleq \bigsqcap_{i = 0}^{n}d_i \sqcap \alpha \circ \bigsqcap_{j = 1}^m e_j$
[i.e., $d \dleq d_i ~ (0 \le i \le n)$ and $d \dleq \alpha \circ e_j ~ (1 \le j \le m)]$
\end{enumerate}
Note that  the qualification value $d$ attached to $\varphi$ is limited by two kinds of upper bounds:
$d_i ~ (0 \le i \le n)$, i.e. the $\simrel$-proximity between $p'(\ntup{t'}{n})$ and the head of $R_l\theta$;
and $\alpha \circ e_j ~ (1 \le j \le m)$, i.e. the qualification values of the atoms in the body of $R_l\theta$
attenuated w.r.t. $R_l$'s attenuation factor $\alpha$.
Moreover, the inequalities $e_j \!\dgeq^? \!\!w_j ~ (1 \le j \le m)$
are required in order to impose the threshold conditions within $R_l$'s body. 
As already explained at the beginning of this subsection, 
$e_j \dgeq^? \!\!w_j$ means that either $w_j =\ ?$ or else $w_j \in \aqdom$ and $e_j \dgeq w_j$.
Now we can define:
\begin{enumerate}
\item
$\I$ is a \emph{model} of a program rule $R_l \in \Prog$
(in symbols, $\I \model{\simrel,\qdom,\cdom} R_l$)
iff every defined observable qc-atom $\varphi$ which is an immediate consequence of $\I$ via $R_l$ verifies $\varphi \in \I$; and
$\I$ is a \emph{model} of $\Prog$ (in symbols, $\I \model{\simrel,\qdom,\cdom} \Prog$)
iff $\I$ is a model of every program rule $R_l \in \Prog$.
\item
$\varphi$ is a \emph{semantic consequence} of $\Prog$ (in symbols, $\Prog \model{\simrel,\qdom,\cdom} \varphi$)
iff $\I \isqchlrdc \varphi$ for every qc-interpretation $\I$ such that $\I \model{\simrel,\qdom,\cdom} \Prog$. \mathproofbox
\end{enumerate}
\end{defn}

The next example may serve as a concrete illustration:

\begin{exmp}[Models and semantic consequence]
\label{exmp:semantic-consequence}
Recall the $\sqclp{\simrel}{\U}{\rdom}$-program $\Prog$ from Example \ref{exmp:running}.
Let us show that the three qc-atoms $\varphi_1$, $\varphi_2$ and $\varphi_3$ from Example \ref{exmp:qc-atoms}
are semantic consequences of $\Prog$:
\begin{enumerate}
  \item
  Assume an arbitrary model $\I \model{\simrel,\U,\rdom} \Prog$.
  Note that the atom underlying $\varphi_1$ and the head atom of $R_1$ are $q(X,c'(Y))$ and $q(X,c(X))$, respectively.
  Since $\simrel(c,c') = 0.9$ and $\Pi \model{\cdom} X == Y$, $\varphi_1$ can be obtained as an immediate consequence of $\I$ via $R_1$ using $\theta = \varepsilon$.
  Therefore $\varphi_1 \in \I$ and we can conclude that $\Prog \model{\simrel,\U,\rdom} \varphi_1$.
   \item
   Assume an arbitrary model $\I \model{\simrel,\U,\rdom} \Prog$.
   Consider the substitution $\theta = \{Y \mapsto c'(Y)\}$.
   Note that the atom underlying $\varphi_2$ and the head atom of $R_2\theta$ are $p'(c'(Y),c(X))$ and $p(c(X),c'(Y))$, respectively.
   Moreover, $\varphi_1 \in \I$ (due to the previous item) and the atom $q(X,c'(Y))$ underlying $\varphi_1$
  is the same as the atom in the body of $R_2\theta$.
  These facts together with $\simrel(p,p') = 0.8$, $\simrel(c,c') = 0.9$  and $\Pi \model{\cdom} X == Y$
  allow to obtain $\varphi_2$ as an immediate consequence of $\I$ via $R_2$.
  Therefore $\varphi_2 \in \I$ and we can conclude that $\Prog \model{\simrel,\U,\rdom} \varphi_2$.
  \item
  Assume an arbitrary model $\I \model{\simrel,\U,\rdom} \Prog$.
   Consider again the substitution $\theta = \{Y \mapsto c'(Y)\}$.
   Note that the atom underlying $\varphi_3$ and the head atom of $R_3\theta$ are $r(c'(Y),c(X),Z)$ and $r(c(X),c'(Y),Z)$, respectively.
   Moreover, the two annotated atoms $\qat{B_j\theta}{w_j} ~ (1 \le j \le 2)$ occurring in the body of $R_3\theta$ are such that
   $\I \isqchlrdc \cqat{B_j\theta}{e_j}{\Pi}$ for suitable values $e_j \geq^? w_j$, namely $e_1 = 0.9$ and $e_2 = 1.0$.
   Note that $e_1 = 0.9$ works because $B_1\theta$ is the atom $q(X,c'(Y))$ underlying  $\varphi_1$
   and $\varphi_1 \in \I$, as proved in the first item of this example. On the  other hand, $e_2 = 1.0$ works because
   $B_2\theta$ is the primitive atom $cp_{\geq}(X,0.0)$ which is trivially entailed by $\Pi$.
   All these facts, together with $\simrel(c,c') = 0.9$, $0.8 \leq 0.9 \times 0.9$  and $\Pi \model{\cdom} X == Y$
   allow to obtain $\varphi_3$ as an immediate consequence of $\I$ via $R_3$.
  Therefore $\varphi_3 \in \I$ and we can conclude that $\Prog \model{\simrel,\U,\rdom} \varphi_3$.  \enspace  \mathproofbox
\end{enumerate}
\end{exmp}

Now we are ready to obtain results on the declarative semantics of programs in the $\mbox{SQCLP}$ scheme.
We will characterize the observable consequences of a given program $\Prog$ in two different, but equivalent, ways:
either using the interpretation transformer presented in Subsection \ref{sec:sqclp:fixpointsem},
or using the extension of Horn Logic presented in Subsection \ref{sec:sqclp:sqchl}.
In both approaches, we will prove the existence of a least model $\Mp$ for each given program $\Prog$.

\subsubsection{A Fixpoint Semantics}
\label{sec:sqclp:fixpointsem}

A well-known way of characterizing models and least models of programs in declarative languages
proceeds by considering a lattice structure for the family of all program interpretations,
and using an interpretation transformer to compute the immediate consequences obtained from program rules.
This kind of approach is well known for logic programming \cite{VEK76,AVE82,Llo87,Apt90}
and constraint logic programming \cite{GL91,GDL95,JMM+98}.
It has been used also in various extensions of logic programming designed to support uncertain reasoning,
such as quantitative logic programming \cite{VE86},
its extension to qualified logic programming \cite{RR08}
quantitative constraint logic programming \cite{Rie96,Rie98phd},
similarity-based logic programming \cite{Ses02}
and proximity-based logic programming in the sense of  \textsf{Bousi}$\sim$\textsf{Prolog}  \cite{JR09}.

The $\mbox{SQCLP}$ scheme  is intended to unify all these logic programming extensions in a common framework.
This subsection is based on the declarative semantics given in \cite{RR08,RR08TR},
extended to deal with constraints and proximity relations.
Our first result provides a lattice of program interpretations.


\begin{prop}[Lattice of Interpretations]
\label{prop:intdc-lattice}
$\Intdc$, defined as the set of all qc-interpretations over the qualification domain $\qdom$ and the constraint domain $\cdom$, is a complete lattice w.r.t. the set inclusion ordering $\subseteq$. Moreover, the bottom element $\ibot$ and the top element $\itop$ of this lattice are characterized as $\ibot = \emptyset$ and $\itop = \{\varphi \mid \varphi \mbox{ is a defined observable qc-atom}\}$ and for any subset $I \subseteq \Intdc$ its greatest lower bound (glb) and least upper bound (lub) are characterized as follows:
\begin{enumerate}
\item \label{prop:intdc-lattice:infimum}
The glb of $I$ (written as $\infi I$) is $\inter_{\I \in I} \I$, understood as $\itop$ if $I = \emptyset$; and
\item \label{prop:intdc-lattice:supremum}
The lub of $I$ (written as $\supr I$) is $\union_{\I \in I} \I$, understood as $\ibot$ if $I = \emptyset$.
\end{enumerate}
\end{prop}

\begin{proof*}
Both $\ibot$ and $\itop$ are qc-interpretations because they are sets of defined observable qc-atoms
and they are closed under $(\qdom, \cdom)$-entailment for trivial reasons, namely:
$\ibot$ is empty and $\itop$ includes all the defined observables.
Moreover, they are the minimum and the maximum of $\Intdc$ w.r.t. $\subseteq$
because $\ibot \subseteq \I \subseteq \itop$ is trivially true for each $\I \in \Intdc$.
Thus, we have only left to prove \ref{prop:intdc-lattice:infimum}. and \ref{prop:intdc-lattice:supremum}.:
\begin{enumerate}
\item
$\inter_{\I \in I} \I$ is obviously a set of defined observable qc-atoms because this is the case for each $\I \in I$.
Given any $\varphi \in \inter_{\I \in I}$ and any observable defined qc-atom $\varphi'$ such that $\varphi \entail{\qdom,\cdom} \varphi'$,
we get $\varphi' \in \inter_{\I \in I} \I$ as an obvious consequence of the fact that each $\I \in I$  is closed under $(\qdom, \cdom)$-entailment.
Therefore, $\inter_{\I \in I} \I \in \Intdc$.
Obviously,  $\inter_{\I \in I} \I$ is trivially a lower bound of $I$ w.r.t. $\subseteq$.
Moreover,  $\inter_{\I \in I} \I$ is the glb of $I$, because any given lower bound $\J$ of $I$
verifies $\J \subseteq \I$ for every $\I \in I$ and thus $\J \subseteq \inter_{\I \in I} \I$.
Therefore, $\inter_{\I \in I} \I = \infi I$.
\item
Using the properties of the union of a family of sets it is easy to prove that $\union_{\I \in I} \I \in \Intdc$
and also that $\union_{\I \in I} \I$ is the lub of $I$ w.r.t. $\subseteq$.
A more detailed reasoning would be similar to the previous item.
Therefore, $\union_{\I \in I} \I = \supr I$. \mathproofbox
\end{enumerate}
\end{proof*}

Next we define an \emph{interpretation transformer} $\Tp$, intended to compute the immediate consequences obtained from a given qc-interpretation via the program rules belonging to $\Prog$.


\begin{defn}[Interpretations Transformer]
\label{defn:tp}
Let $\Prog$ be a fixed $\sqclp{\simrel}{\qdom}{\cdom}$-program.
The interpretations transformer $\Tp : \Intdc \to \Intdc$ is defined by the condition:
$$\Tp(\I) \eqdef \{ \varphi \mid \varphi \mbox{ is an immediate consequence of } \I \mbox{ via some } R_l \in \Prog \}
\enspace . \mathproofbox$$
\end{defn}

The computation of immediate consequences of a given qc-interpretation $\I$ via a given program rule $R_l$
has been already explained in  Definition \ref{defn:models}. The following example illustrates the workings of $\Tp$.

\begin{exmp}[Interpretation transformer in action]
\label{exmp:tp} Recall again the $\sqclp{\simrel}{\U}{\rdom}$-program $\Prog$ from Example
\ref{exmp:running} and the observable defined qc-atoms $\varphi_1$, $\varphi_2$ and $\varphi_3$
from Example \ref{exmp:qc-atoms}. Then:
\begin{enumerate}
\item
The arguments given in Example \ref{exmp:semantic-consequence}(1)
can be easily reused to show that $\varphi_1$ is an immediate consequence of the empty interpretation $\ibot$ via the program rule $R_1$.
Therefore, $\varphi_1 \in \Tp(\ibot)$.
\item
The arguments given in Example \ref{exmp:semantic-consequence}(2)
can be easily reused to show that $\varphi_1$ is an immediate consequence of $\I$ via the program rule $R_2$,
provided that $\varphi_1 \in \I$. Therefore, $\varphi_2 \in \Tp(\Tp(\ibot))$.
\item
The arguments given in Example \ref{exmp:semantic-consequence}(3)
can be easily reused to show that $\varphi_3$ is an immediate consequence of $\I$ via the program rule $R_3$,
provided that $\varphi_1 \in \I$. Therefore, $\varphi_3 \in \Tp(\Tp(\ibot))$. \mathproofbox
\end{enumerate}
\end{exmp}

The next proposition states the main properties of interpretation transformers.

\begin{prop}[Properties of interpretation transformers]
\label{prop:tp-properties}
Let $\Prog$ be any fixed $\sqclp{\simrel}{\qdom}{\cdom}$-program. Then:
\begin{enumerate}
\item\label{prop:tp-properties:1}
$\Tp$ is a well defined mapping, 
i.e. for all $\I \in \Intdc$ one has  $\Tp(\I) \in \Intdc$.
\item\label{prop:tp-properties:2}
$\Tp$ is monotonic and continuous.
\item\label{prop:tp-properties:3}
For all $\I \in \Intdc$ one has: $\I \model{\simrel,\qdom,\cdom} \Prog \Longleftrightarrow \Tp(\I) \subseteq \I$,
That is, the models of $\Prog$ are precisely the pre-fixpoints of $\Tp$.
\end{enumerate}
\end{prop}

\begin{proof*}
\begin{enumerate}
\item
By definition, $\Tp(\I)$ is a set of observable defined qc-atoms. 
It is sufficient to prove that it is closed under  $(\qdom, \cdom)$-entailment.
Let us assume two observable defined qc-atoms $\varphi$ and $\varphi'$ such that $\varphi \in \Tp(\I)$ 
and $\varphi \entail{\qdom,\cdom} \varphi'$.
Because of $\varphi \entail{\qdom,\cdom} \varphi'$ we can assume $\varphi : \cqat{p(\ntup{t}{n})}{d}{\Pi}$, 
$\varphi' : \cqat{p(\ntup{t'}{n})}{d'}{\Pi'}$
and some substitution $\theta $ such that $p(\ntup{t'}{n}) = p(\ntup{t}{n})\theta$, $d' \dleq d$ and $\Pi' \model{\cdom} \Pi\theta$.
Because of $\varphi \in \Tp(\I)$, we can assume that $\varphi$ is an immediate consequence of $\I$ via some $R_l \in \Prog$.
More precisely, we can assume $(R_l : q(\ntup{s}{n}) \qgets{\alpha} \qat{B_1}{w_1}, \ldots, \qat{B_m}{w_m}) \in \Prog$, some substitution $\sigma$
and some qualification values $d_0, d_1, \ldots,$ $d_n, e_1, \ldots, e_m \in \aqdom$ such that
\begin{enumerate}
\item
$\simrel(p,q) = d_0$,
\item
$\I \isqchlrdc \cqat{(t_i == s_i\sigma)}{d_i}{\Pi}$ for $i = 1 \ldots n$,
\item
$\I \isqchlrdc \cqat{B_j\sigma}{e_j}{\Pi}$ with $e_j \dgeq^? w_j$ for $j = 1 \ldots m$,
\item
$d \dleq \bigsqcap_{i = 0}^{n}d_i \sqcap \alpha \circ \bigsqcap_{j = 1}^m e_j$ 
[i.e., $d \dleq d_i ~ (0 \le i \le n)$ and $d \dleq \alpha \circ e_j ~ (1 \le j \le m)$].
\end{enumerate}
In order to show that $\varphi' \in \Tp(\I)$, we claim that $\varphi'$ can be computed as an immediate consequence of $\I$
via the same program rule $R_l$, using the substitution $\sigma\theta$ and the qualification values
$d_0, d_1, \ldots,$ $d_n, e_1, \ldots, e_m \in \aqdom$. To justify this claim it is enough to check the following items:
\begin{enumerate}
\item[(a')]
$\simrel(p,q) = d_0$,
\item[(b')]
$\I \isqchlrdc \cqat{(t'_i == s_i\sigma\theta)}{d_i}{\Pi'}$ for $i = 1 \ldots n$,
\item[(c')]
$\I \isqchlrdc \cqat{B_j\sigma\theta}{e_j}{\Pi'}$ with $e_j \dgeq^? w_j$ for $j = 1 \ldots m$,
\item[(d')]
$d \dleq \bigsqcap_{i = 0}^{n}d_i \sqcap \alpha \circ \bigsqcap_{j = 1}^m e_j$ 
[i.e., $d \dleq d_i ~ (0 \le i \le n)$ and $d \dleq \alpha \circ e_j ~ (1 \le j \le m)$].
\end{enumerate}
These four items closely correspond to items (a)-(d) above. More specifically: \\
--- Items (a') and (d') are identical to items (a) and (d), respectively. \\
--- Regarding item (b'):
For $i = 1 \ldots n$,  $\I \isqchlrdc \cqat{(t'_i == s_i\sigma\theta)}{d_i}{\Pi}$ is the same as
$t_i\theta \approx_{d_i, \Pi'} s_i\sigma\theta$.
Because of Lemma \ref{lem:slpeq}, this is a consequence of $\Pi' \model{\cdom} \Pi\theta$
and $t_i \approx_{d_i, \Pi} s_i\sigma$,
which is ensured by item (b). \\
--- Regarding  item (c'):
For $j = 1 \ldots m$, $e_j \dgeq^? w_j$ is ensured by item (c),
and $\I \isqchlrdc \cqat{B_j\sigma\theta}{e_j}{\Pi'}$ follows from $\I \isqchlrdc \cqat{B_j\sigma}{e_j}{\Pi}$
--also ensured by item (c)--
and the entailment property for interpretations (Proposition \ref{prop:ep-i}),
which can be applied because $\cqat{B_j\sigma}{e_j}{\Pi} \entail{\qdom,\cdom} \cqat{B_j\sigma\theta}{e_j}{\Pi'}$.

\item
Monotonicity means that the inclusion  $\Tp (\I) \subseteq \Tp(\J)$ holds whenever $\I \subseteq \J$.
This follows very easily from
$$(\spadesuit) \quad \I \isqchlrdc \varphi \mbox{ and } \I \subseteq \J \Longrightarrow \J \isqchlrdc \varphi$$
which is a trivial consequence of Definition \ref{defn:interpretations}.

Continuity means that the equation $\Tp(\supr I) = \supr \{\Tp(\I) \mid \I \in I \}$ holds for any directed set $I \subseteq \Intdc$ of qc-interpretations.
Recall that $I \subseteq \Intdc$ is called directed iff every finite subset $I_0 \subseteq I$ has some upper bound $\I \in I$.
We show that $\Tp(\supr I) = \supr \{Tp(\I) \mid \I \in I \}$ holds by proving the  two inclusions separately:

\begin{enumerate}
\item
For each fixed $\I_0 \in I$, $\Tp(\I_0) \subseteq  \Tp(\supr I)$ follows from $\I_0 \subseteq \supr I$ and monotonicity of $\Tp$.
Then, the inclusion $\supr \{\Tp(\I) \mid \I \in I \} \subseteq \Tp(\supr I)$ holds by definition of supremum.
\item
In order to prove the opposite inclusion $\Tp(\supr I) \subseteq \supr \{\Tp(\I) \mid \I \in I \}$,
consider an arbitrary $\varphi \in \Tp(\supr I)$. Due to Definition \ref{defn:tp}, $\varphi$ is an
immediate consequence of $\supr I$ via some program rule $R_l \in \Prog$. Because of the first item
of Definition \ref{defn:models}, $\varphi$ is an immediate consequence of $\supr I$ via $R_l$ due
to finitely many qc-facts of the form $\cqat{B_j\theta}{e_j}{\Pi}$ (coming from the body of a
suitable instance of $R_l$) that are valid in $\supr I$. Because of $(\spadesuit)$ and the
assumption that  $I$ is a directed set,  it is possible to choose some $\I_0 \in I$ such that all
the qc-facts $\cqat{B_j\theta}{e_j}{\Pi}$ are valid in $\I_0$. Then, $\varphi$ is an immediate
consequence of this particular $\I_0 \in I$ via $R_l$. Therefore, $\varphi \in \Tp(\I_0) \subseteq
\supr \{\Tp(\I) \mid \I \in I \}$.
\end{enumerate}

\item
According to Definition \ref{defn:models}, $\I \model{\simrel,\qdom,\cdom} \Prog$ holds iff
every observable defined qc-atom $\varphi$ which is an immediate consequence of $\I$ via the program rules $R_l \in \Prog$ verifies $\varphi \in \I$.
According to Definition \ref{defn:tp}, $\Tp(\I)$ is just the set of all the defined observable qc-atoms $\varphi$
that can be obtained as immediate consequences of $\I$ via the program rules $R_l \in \Prog$.
Consequently, $\I \model{\simrel,\qdom,\cdom} \Prog$ holds iff $\Tp(\I) \subseteq \I$.  \enspace \mathproofbox
\end{enumerate}
\end{proof*}

The theorem below is the main result in this subsection.

\begin{thm}[Fixpoint characterization of least program models]
\label{thm:tp-leastmodel}
Every $\sqclp{\simrel}{\qdom}{\cdom}$-program $\Prog$ has a \emph{least model} $\Mp$,
smaller than any other model of $\Prog$ w.r.t. the set inclusion ordering of the interpretation lattice $\Intdc$.
Moreover, $\Mp$ can be characterized as the {\em least fixpoint} of $\Tp$ as follows:
$$\Mp = l\!f\!p(\Tp) = \union_{k\in\NAT} \Tp{\uparrow^k}(\ibot) \enspace . \mathproofbox$$
\end{thm}

\begin{proof}
As usual, a given $\I \in \Intdc$ is called a fixpoint of $\Tp$ iff $\Tp(\I) = \I$,
and $\I$ is called a pre-fixpoint of $\Tp$ iff $\Tp(\I) \subseteq \I$.
Due to a well-known theorem by Knaster and Tarski, see \cite{Tar55}, a monotonic mapping from a complete lattice into itself always has a least fixpoint which is also its least pre-fixpoint. In the case that the mapping is continuous, its least fixpoint can be characterized as the lub of the sequence of lattice elements obtained by reiterated application of the mapping to the bottom element. Combining these results with Proposition \ref{prop:tp-properties} trivially proves the theorem.
\end{proof}

\subsubsection{An equivalent Proof-theoretic Semantics}
\label{sec:sqclp:sqchl}


In order to give a logical view of program semantics and an alternative characterization of least program models,
we define the \emph{Proximity-based Qualified Constrained Horn Logic} $\SQCHL(\simrel,\qdom,\cdom)$
as a formal inference system consisting of the three inference rules displayed in Figure \ref{fig:sqchl}.

\begin{figure}[h]
  \figrule
  \centering
  \begin{tabular}{ll}&\\
  \textbf{SQDA} & $\displaystyle\frac
    {~ (~ \cqat{(t'_i == t_i\theta)}{d_i}{\Pi} ~)_{i=1 \ldots n} \quad (~ \cqat{B_j\theta}{e_j}{\Pi} ~)_{j=1 \ldots m} ~}
    {\cqat{p'(\ntup{t'}{n})}{d}{\Pi}}$ \\ \\
  &if $(p(\ntup{t}{n}) \qgets{\alpha} \qat{B_1}{w_1}, \ldots, \qat{B_m}{w_m}) \in \Prog$, $\theta$ subst.,
  $\simrel(p',p) = d_0 \neq \bt$,\\ 
  &$e_j \dgeq^? w_j ~ (1 \le j \le m)$ and $d \dleq \bigsqcap_{i = 0}^{n}d_i \sqcap \alpha \circ \bigsqcap_{j = 1}^m e_j$.\\
  &\\
  \textbf{SQEA} & $\displaystyle\frac
    {}
    {\quad \cqat{(t == s)}{d}{\Pi} \quad}$ ~
  if $t \approx_{d, \Pi} s$. \\
  &\\
  \textbf{SQPA} & $\displaystyle\frac
    {}
    {\quad \cqat{\kappa}{d}{\Pi} \quad}$ ~
  if $\Pi \model{\cdom} \kappa$. \\
  &\\
  \end{tabular}
  \caption{Proximity-based Qualified Constrained Horn Logic}
  \label{fig:sqchl}
  \figrule
\end{figure}


The three inference rules are intended to work with observable qc-atoms.
Rule \textbf{SQDA} is used to infer defined qc-atoms.
It formalizes an extension of  the classical \emph{Modus Ponens}  inference,
allowing  to infer a defined qc-atom $\cqat{p'(\ntup{t'}{n})}{d}{\Pi}$ by means of an instance of a program clause
with head $p(\ntup{t}{n})\theta$ and body atoms $\qat{B_j\theta}{w_j}$.
The $n$ premises $\cqat{(t'_i == t_i\theta)}{d_i}{\Pi}$ combined with the
side condition $\simrel(p',p) = d_0 \neq \bt$  ensure the ``equality'' between
$p'(\ntup{t'}{n})$ and $p(\ntup{t}{n})\theta$ modulo $\simrel$;
the $m$ premises $\cqat{B_j\theta}{e_j}{\Pi}$ require to prove the body atoms;
and the side conditions $e_j \dgeq^? w_j$ and $d \dleq \bigsqcap_{i =
0}^{n}d_i \sqcap \alpha \circ \bigsqcap_{j = 1}^m e_j$ check the threshold conditions of the body
atoms and impose the proper relationships between the qualification
value attached to  the conclusion and  the qualification values attached to the premises.
In particular, the inequality $d \dleq \alpha \circ \bigsqcap_{j = 1}^m e_j$ is imposed,
meaning that the qualification value attached to a clause's head cannot exceed the glb
of the qualification values attached to the body atoms attenuated by the clause's attenuation factor.
Rules \textbf{SQEA} and \textbf{SQPA} are used to infer equational and primitive qc-atoms, respectively.
Rule \textbf{SQEA} is designed to work with term proximity w.r.t. $\Pi$ in the sense of 
Definition \ref{defn:Pi-prox}, inferring  $\cqat{(t  == s)}{d}{\Pi}$ just in the case that $t \approx_{d, \Pi} s$ holds.
Rule  \textbf{SQPA} infers $\cqat{\kappa}{d}{\Pi}$ for an arbitrary $d \in \aqdom$, provided that $\Pi \model{\cdom} \kappa$ holds. This makes sense because the requirements for admissible triples in Definition \ref{defn:simrel:admissible} include the assumption that 
$\simrel(p,p') \neq \bt$ cannot happen if $p, p'  \in PP$ are syntactically different primitive predicate symbols. 

As usual in formal inference systems, $\SQCHL(\simrel,\qdom,\cdom)$ proofs can be represented as {\em proof trees} $T$ whose nodes correspond to qc-atoms, each node being  inferred from its children by means of some  $\SQCHL(\simrel,\qdom,\cdom)$ inference step. In the rest of the report we will use the following notations:

\begin{itemize}
\item
$\Vert T \Vert$ will denote the {\em size} of  the proof tree $T$,  measured as its number of nodes,
which equals   the number of inference steps in the $\SQCHL(\simrel,\qdom,\cdom)$ proof represented by $T$.
\item
$\Vert T \Vert_d$ will denote the number of nodes of  the proof tree $T$ that represent conclusions of
\textbf{SQDA} inference steps. Obviously, $\Vert T \Vert_d \leq \Vert T \Vert$.
\item
$\Prog \sqchlrdc \varphi$ will indicate that $\varphi$ can be inferred
from $\Prog$ in $\SQCHL(\simrel,\qdom,\cdom)$.
\item
$\Prog \sqchlrdcn{k} \varphi$ will indicate that $\varphi$ can be inferred from $\Prog$ in $\SQCHL(\simrel,\qdom,\cdom)$ using some proof tree $T$ such that $\Vert T \Vert_d = k$.
\end{itemize}

The next example  shows a $\SQCHL(\simrel,\U,\rdom)$ proof tree.


\begin{exmp}[$\SQCHL(\simrel,\qdom,\cdom)$ proof tree]
\label{exmp:sqchl-inference}
Recall the proximity relation $\simrel$ and the program  $\Prog$ from our running Example \ref{exmp:running},
as well as the observable qc-statement $\varphi_2 = \cqat{p'(c'(Y),c(X))}{0.8}{\Pi}$ already known from Example \ref{exmp:qc-atoms}.
A $\SQCHL(\simrel,\U,\rdom)$ proof tree witnessing $\Prog \sqchl{\simrel}{\U}{\rdom} \varphi_2$ can be displayed as follows:

$$
\spadesuit = \displaystyle\frac
{~
  \displaystyle\frac{}{\cqat{(Y == Y)}{1.0}{\Pi}} ~ (5) \qquad
  \displaystyle\frac{}{\cqat{(c(X) == c(Y))}{1.0}{\Pi}} ~ (6)
~}
{\cqat{q(Y,c(X))}{1.0}{\Pi}} ~ (4)
$$

$$
\displaystyle\frac
{~
  \displaystyle\frac{}{\cqat{(c'(Y) == c(Y))}{0.8}{\Pi}} ~ (2) \qquad
  \displaystyle\frac{}{\cqat{(c(X) == c(X))}{1.0}{\Pi}} ~ (3) \qquad
  \spadesuit ~ (4)
~}
{\cqat{p'(c'(Y),c(X))}{0.8}{\Pi}} ~ (1)
$$

\smallskip
The inference steps in this proof are commented below.
For the sake of clarity, we have used a different variant of the corresponding program clause
for each each application of the inference rule  \textbf{SQDA}.
\begin{enumerate}
  \item[(1)]
  \textbf{SQDA} step with clause $R_1 = (~ p(c(X_1),Y_1) \qgets{0.9} q(X_1,Y_1) ~)$
  instantiated by  substitution $\theta_1 = \{ X_1 \mapsto Y, Y_1 \mapsto c(X) \}$.
  Note that $0.8$ satisfies $0.8 \le \simrel(p,p') = 0.8$, $0.8 \le 0.8$, $0.8 \le 1.0$, $0.8 \le 0.9 \times 1.0$.
  \item[(2)]
  \textbf{SQEA} step. $c'(Y) \approx_{0.8, \Pi} c(Y)$ holds due to
  $c'(Y) \approx_{\Pi} c'(Y)$, $c(Y) \approx_{\Pi} c(Y)$ and $c'(Y) \approx_{0.8} c(Y)$.
   \item[(3)]
   \textbf{SQEA} step. $c(X) \approx_{1.0, \Pi} c(X)$ holds for trivial reasons.
  \item[(4)]
  \textbf{SQDA} step with clause $R_2 = (~ q(X_2,c(X_2)) \qgets{1.0} ~)$
  instantiated by substitution $\theta_2 = \{ X_2 \mapsto Y \}$.
  Note that $1.0$ satisfies $1.0 \le \simrel(q,q) = 1.0$ and $1.0 \le 1.0$.
  \item[(5)]
  \textbf{SQEA} step. $Y \approx_{1.0, \Pi} Y$ holds for trivial reasons.
  \item[(6)]
  \textbf{SQEA} step. $c(X) \approx_{1.0, \Pi} c(Y)$ holds due to
  $c(X) \approx_{\Pi} c(Y)$ (which follows from $\Pi \model{\rdom} X == Y$)
  and $c(X) \approx_{1.0} c(X)$.
  \enspace \mathproofbox
\end{enumerate}

\end{exmp}

The next technical lemma establishes two basic properties of formal inference in the $\SQCHL(\simrel,\qdom,\cdom)$ logic.

\begin{lem}[Properties of $\SQCHL(\simrel,\qdom,\cdom)$ derivability]
\label{lem:ep}
Let $\Prog$ be any $\sqclp{\simrel}{\qdom}{\cdom}$-program. Then:
\begin{enumerate}
\item \label{lem:ep:1}
\emph{$\Prog$-independent Inferences}: \\ 
Given any $\cdom$-based qc-atom $\varphi$ and any qc-interpretation $\I$, one has:
$$\Prog  \sqchlrdcn{0} \varphi \Longleftrightarrow \Prog \sqchlrdc \varphi \Longleftrightarrow \I \isqchlrdc \varphi \enspace .$$
\item \label{lem:ep:2}
\emph{Entailment Property  for Programs}: \\ 
Given any pair of qc-atoms $\varphi$ and $\varphi'$ such that $\Prog \sqchlrdc \varphi$ with inference proof tree $T$ and $\varphi \entail{\qdom,\cdom} \varphi'$, then $\Prog \sqchlrdc \varphi'$ with an inference proof tree $T'$ of the same size and structure as $T$.
\end{enumerate}
\end{lem}

\begin{proof*}[Proof of $\Prog$-independent Inferences]
Since $\varphi$ is $\cdom$-based, we can assume $\varphi = \cqat{A}{d}{\Pi}$ where $A$ is either an equation or a primitive atom.
In both cases the equivalence $\Prog  \sqchlrdcn{0} \varphi \Longleftrightarrow \Prog \sqchlrdc \varphi$ is obvious.
In order to prove the equivalence $\Prog \sqchlrdc \varphi \Longleftrightarrow \I \isqchlrdc \varphi$ we distinguish the two cases:
\begin{enumerate}
  \item
  $\varphi$ is equational.
  Then $A$ has the form $t == s$.
  Considering  the $\SQCHL(\simrel,\qdom,\cdom)$-inference rule \textbf{SQEA}
  and the second item of Definition \ref{defn:interpretations}, we get
  $$\Prog \sqchlrdc \varphi \Longleftrightarrow s \approx_{d, \Pi} t \Longleftrightarrow \I \isqchlrdc \varphi \enspace .$$
  \item
  $\varphi$ is primitive.
  Then $A$ is a primitive atom $\kappa$.
  Considering the $\SQCHL(\simrel,\qdom,\cdom)$-inference rule \textbf{SQPA}
  and the second item of Definition \ref{defn:interpretations}, we get
  $$\Prog \sqchlrdc \varphi \Longleftrightarrow \Pi \model{\cdom} \kappa \Longleftrightarrow \I \isqchlrdc \varphi \enspace . \mathproofbox$$
\end{enumerate}
\end{proof*}

\begin{proof*}[Proof of Entailment Property  for Programs]
Due to the hypothesis   $\varphi \entail{\qdom,\cdom} \varphi'$ and Definition \ref{defn:atoms-entail},
we can assume $\varphi = \cqat{A}{d}{\Pi}$ and $\varphi' = \cqat{A'}{d'}{\Pi'}$
with $A' = A\theta$, $d' \dleq d$ and $\Pi' \model{\cdom} \Pi\theta$ for some substitution $\theta$.
We reason by complete induction on $\Vert T \Vert$.
There are three possible cases,  according to the the syntactic form of the atom $A$.
In each case we argue how to build the desired proof tree $T'$.
\begin{enumerate}
\item $A$ is a defined atom:
In this case, $A$ is $p(\ntup{t}{n})$ with $p \in DP^n$, and $A'$: is  $p(\ntup{t'}{n})$ with $p(\ntup{t'}{n}) = p(\ntup{t}{n})\theta$.
Moreover, $T$ must be a proof tree of the following  form:
$$
T :
\displaystyle\frac
  {~
    \left( \displaystyle\frac{}{\cqat{(t_i == s_i\sigma)}{d_i}{\Pi}} \right)_{i=1 \ldots n} \quad
    \left( \displaystyle\frac{\cdots}{\cqat{B_j\sigma}{e_j}{\Pi}} \right)_{j=1 \ldots m}
  ~}
  {\cqat{p(\ntup{t}{n})}{d}{\Pi}} ~ \mathbf{SQDA}
$$
where:
  \begin{itemize}
  \item
  The \textbf{SQDA} root inference uses
  some $R_l : (q(\ntup{s}{n}) \qgets{\alpha} \qat{B_1}{w_1}, \ldots, \qat{B_m}{w_m}) \in \Prog$,
  some substitution $\sigma$
  and some qualification values $d_0, d_1, \ldots, d_n, e_1, \ldots e_m \in \aqdom$ such that
  $\simrel(p,q) = d_0 \neq \bt$, $d \dleq d_i ~ (0 \le i \le n)$ and $d \dleq \alpha \circ e_j ~ (1 \le j \le m)$.
  \item
  For $i = 1 \ldots n$, $\cqat{(t_i == s_i\sigma)}{d_i}{\Pi}$ has a proof tree $T_{i}^h$
  with $\Vert T_{i}^h \Vert < \Vert T \Vert$.
  \item
  For $j = 1 \ldots m$, $\cqat{B_j\sigma}{e_j}{\Pi}$ has a proof tree $T_{j}^b$
  with $\Vert T_{j}^b \Vert < \Vert T \Vert$.
  \end{itemize}
Then, $T'$ can be built as a proof tree of the form:
$$
T' :
\displaystyle\frac
  {~
    \left( \displaystyle\frac{}{\cqat{(t'_i == s_i\sigma\theta)}{d_i}{\Pi'}} \right)_{i=1 \ldots n} \quad
    \left( \displaystyle\frac{\cdots}{\cqat{B_j\sigma\theta}{e_j}{\Pi'}} \right)_{j=1 \ldots m}
  ~}
  {\cqat{p(\ntup{t'}{n})}{d'}{\Pi'}} ~ \mathbf{SQDA}
$$
where:
  \begin{itemize}
  \item
  The \textbf{SQDA} root inference uses the same program clause $R_l \in \Prog$,
  the substitution $\sigma\theta$
  and the same qualification values $d_i ~ (0 \le i \le n)$ and $e_j ~ (1 \le j \le m)$,
  satisfying $\simrel(p,q) = d_0 \neq \bt$, $d' \dleq d \dleq d_i ~ (0 \le i \le n)$ and $d' \dleq d \dleq \alpha \circ e_j ~ (1 \le j \le m)$.
  \item
  For $i = 1 \ldots n$, $\cqat{(t'_i == s_i\sigma\theta)}{d_i}{\Pi'}$ has a proof tree $T_{i}^{'h}$
  of the same size and structure as $T_{i}^h$.
  In fact, $T_{i}^{'h}$ can be obtained by induction hypothesis applied to $T_{i}^h$,
  which is allowed because $\Vert T_{i}^h \Vert < \Vert T \Vert$
  and $\cqat{(t_i == s_i\sigma)}{d_i}{\Pi} \entail{\qdom,\cdom} \cqat{(t'_i == s_i\sigma\theta)}{d_i}{\Pi'}$.
  Note that this entailment holds thanks to substitution $\theta$, since $t'_i = t_i\theta$ and $\Pi' \model{\cdom} \Pi\theta$.
  \item
  For $j = 1 \ldots m$, $\cqat{B_j\sigma\theta}{e_j}{\Pi'}$ has a proof tree $T_{j}^{'b}$
  of the same size and structure as $T_{j}^b$.
  In fact, $T_{j}^{'b}$ can be obtained by induction hypothesis applied to $T_{j}^b$,
  which is allowed because $\Vert T_{j}^b \Vert < \Vert T \Vert$
  and $\cqat{B_j\sigma}{e_j}{\Pi} \entail{\qdom,\cdom} \cqat{B_j\sigma\theta}{e_j}{\Pi'}$.
  Note that this entailment holds thanks to substitution $\theta$, since $\Pi' \model{\cdom} \Pi\theta$.
  \end{itemize}
By construction, $T'$ has the same size and structure as $T$, as desired.
\item $A$ is an equation:
In this case, $A : t == s$ and $A' : t' == s'$ with $t' = t\theta$, $s' = s\theta$.
Moreover, $T$ must consist of one single node $\cqat{(t == s)}{d}{\Pi}$ inferred by means of \textbf{SQEA}.
Therefore, $t \approx_{d, \Pi} s$ holds.
This implies $t\theta \approx_{d, \Pi\theta} s\theta$ (i.e. $t' \approx_{d, \Pi\theta} s'$)
due to the Substitution Lemma \ref{lem:slpeq}.
From this we conclude $t' \approx_{\Pi'} s'$ due to $d' \dleq d$ and $\Pi' \model{\cdom} \Pi\theta$.
Therefore, $T'$ can be built as a proof tree consisting of one single node $\cqat{(t' == s')}{d'}{\Pi'}$ inferred by means of \textbf{SQEA}.
\item $A$ is a primitive atom:
In this case, $A : \kappa$ and $A' : \kappa' = \kappa\theta$.
Moreover, $T$ must consist of one single node $\cqat{\kappa}{d}{\Pi}$ inferred by means of \textbf{SQPA}.
Therefore, $\Pi \model{\cdom} \kappa$ holds.
This implies $\Pi\theta \model{\cdom} \kappa\theta$ due to the Substitution Lemma \ref{lem:sl}.
From this we conclude $\Pi' \model{\cdom} \kappa'$ due to $\kappa' = \kappa\theta$ and $\Pi' \model{\cdom} \Pi\theta$.
Therefore, $T'$ can be built as a proof tree consisting of one single node $\cqat{\kappa'}{d'}{\Pi'}$ inferred by means of \textbf{SQPA}.
\mathproofbox
\end{enumerate}
\end{proof*}

The following theorem is the main result in this subsection.
It  characterizes the least model of a $\sqclp{\simrel}{\qdom}{\cdom}$-program $\Prog$ w.r.t. the logic $\SQCHL(\simrel,\qdom,\cdom)$:

\begin{thm}[Logical characterization of least program models]
\label{thm:SQCHL-leastmodel}
For any $\sqclp{\simrel}{\qdom}{\cdom}$-program $\Prog$, its least model can be characterized as: $$\Mp = \{\varphi \mid \varphi \mbox{ is a defined observable qc-atom and }\Prog \sqchlrdc \varphi\} \enspace .$$
\end{thm}
\begin{proof*}
By Theorem \ref{thm:tp-leastmodel}, we already know that $\Mp = \union_{k\in\NAT} \Tp{\uparrow}^k(\ibot)$.
Therefore, it is sufficient to prove that the two implications
\begin{enumerate}
\item \label{thm:SQCHL-leastmodel:1} 
$\Prog \sqchlrdcn{k} \varphi \Longrightarrow \exists k' : \varphi \in \Tp{\uparrow}^{k'}(\ibot)$
\item \label{thm:SQCHL-leastmodel:2} 
$\varphi \in \Tp{\uparrow}^k(\ibot) \Longrightarrow \exists k' : \Prog \sqchlrdcn{k'} \varphi$
\end{enumerate}
hold for any defined observable qc-atom $\varphi = \cqat{p(\ntup{t}{n})}{d}{\Pi}$ and for any integer value $k \geq 1$.
We prove both implications within one single inductive reasoning on $k$.

\begin{description}
\item[Basis ($k = 1$).] \hfill \\
--- {\it Implication 1.}
Assume $\Prog \sqchlrdcn{1} \varphi$.
Then, due to the single \textbf{SQDA} inference,
there must exist  some $R_l = (q(\ntup{s}{n}) \qgets{\alpha}) \in \Prog$ with empty body,
some substitution $\theta$ and some $d_0, d_1, \ldots, d_n \in \aqdom$
such that $\Prog \sqchlrdcn{0} \cqat{(t_i == s_i\theta)}{d_i}{\Pi}$ for $i = 1 \ldots n$,
$\simrel(p,q) = d_0 \neq \bt$, $d \dleq d_i ~ (0 \le i \le n)$ and $d \dleq \alpha$.
Then $\ibot \sqchlrdcn{0} \cqat{(t_i == s_i\theta)}{d_i}{\Pi}$ holds for $i = 1 \ldots n$,
because of Lemma \ref{lem:ep}(\ref{lem:ep:1}).
Therefore $\varphi$ is an immediate consequence of $\ibot$ via $R_l$,
which guarantees $\varphi \in \Tp{\uparrow}^1(\ibot)$. \\
--- {\it Implication 2.}
Assume now $\varphi \in \Tp{\uparrow}^1(\ibot)$.
Then $\varphi$ must be an immediate consequence of $\ibot$
via some $R_l = (q(\ntup{s}{n}) \qgets{\alpha}) \in \Prog$ with empty body.
Then there are some substitution $\theta$ and some $d_0, d_1, \ldots, d_n \in \aqdom$
such that $\ibot \isqchlrdc \cqat{(t_i == s_i\theta)}{d_i}{\Pi}$ for $i = 1 \ldots n$,
$\simrel(p,q) = d_0 \neq \bt$, $d \dleq d_i ~ (0 \le i \le n)$ and $d \dleq \alpha$.
Again because of Lemma \ref{lem:ep}(\ref{lem:ep:1}),
we get $\Prog \sqchlrdcn{0} \cqat{(t_i == s_i\theta)}{d_i}{\Pi}$ for $i = 1 \ldots n$,
which guarantees $\Prog \sqchlrdcn{1} \varphi$
with one single \textbf{SQDA} inference using $R_l$ instantiated by $\theta$.

\smallskip
\smallskip
\item[Inductive step ($k > 1$).] \hfill \\
--- {\it Implication 1.} Assume $\Prog \sqchlrdcn{k} \varphi$. Since the root inference must be \textbf{SQDA},  there must exist  some program rule $(R_l : q(\ntup{s}{n}) \qgets{\alpha} \qat{B_1}{w_1}, \ldots, \qat{B_m}{w_m}) \in \Prog$, some substitution $\theta$ and some qualification values $d_0, d_1$, \ldots, $d_n, e_1, \ldots, e_m \in \aqdom$ such that
      \begin{itemize}
        \item $\Prog \sqchlrdcn{0} \phi_i = (\cqat{(t_i == s_i\theta)}{d_i}{\Pi})$ for $i = 1 \ldots n$,
        \item $\Prog \sqchlrdcn{k_j} \psi_j = (\cqat{B_j\theta}{e_j}{\Pi})$ with $e_j \dgeq^? w_j$ for $j = 1 \ldots m$, and
        \item $\simrel(p,q) = d_0 \neq \bt$, $d \dleq d_i ~ (0 \le i \le n)$ and $d \dleq \alpha \circ e_j ~ (1 \le j \le m)$
      \end{itemize}
      where $\Sigma_{j = 1}^{m} k_j = k - 1$. For each $j = 1 \ldots m$,
      either $\psi_j$ is defined, and then induction hypothesis yields some $k'_j$ such that
       $\psi_j \in \Tp{\uparrow}^{k'_j}(\ibot)$ and therefore also $\Tp{\uparrow}^{k'_j}(\ibot) \isqchlrdc \psi_j$;
      or else $\psi_j$ is not defined and then $\Tp{\uparrow}^{k'_j}(\ibot) \isqchlrdc \psi_j$
      for any arbitrarily chosen $k'_j$, by Lemma \ref{lem:ep}(\ref{lem:ep:1}).
      Then $l = max \{k'_j \mid 1 \leq j \leq m\}$ verifies that $\varphi$ is an immediate consequence of
      $\Tp{\uparrow}^{l}(\ibot)$ via $R_l$, which implies $\varphi \in \Tp{\uparrow}^{k'}(\ibot)$ for $k' = l+1$. \\
--- {\it Implication 2.} Assume $\varphi \in \Tp{\uparrow}^k(\ibot) = \Tp(\Tp{\uparrow}^{k-1}(\ibot))$.
Then $\varphi$ is an immediate consequence of $Tp{\uparrow}^{k-1}(\ibot)$
via some clause $(R_l : q(\ntup{s}{n})\qgets{\alpha} \qat{B_1}{w_1}, \ldots,$ $\qat{B_m}{w_m}) \in \Prog$.
Therefore, there exist some substitution $\theta$ and some qualification values $d_0, d_1$, \ldots, $d_n, e_1, \ldots, e_m$ $\in$ $\aqdom$ such that:
      \begin{itemize}
        \item $\Tp{\uparrow}^{k{-}1}(\ibot) \isqchlrdc \phi_i = (\cqat{(t_i == s_i\theta)}{d_i}{\Pi})$ for $i = 1 \ldots n$,
        \item $\Tp{\uparrow}^{k{-}1}(\ibot) \isqchlrdc \psi_j = (\cqat{B_j\theta}{e_j}{\Pi})$ with $e_j \dgeq^? w_j$ for $j = 1 \ldots m$, and
        \item $\simrel(p,q) = d_0 \neq \bt$, $d \dleq d_i ~ (0 \le i \le n)$ and $d \dleq \alpha \circ e_j ~ (1 \le j \le m)$.
      \end{itemize}
For each $i = 1 \ldots n$, Lemma \ref{lem:ep}(\ref{lem:ep:1}) yields $\Prog \sqchlrdcn{0} \phi_i$.
For each $j = 1 \ldots m$,
either $\psi_j$ is defined,  in which case $\psi_j \in \Tp{\uparrow}^{k{-}1}(\ibot)$, $k-1 \geq 1$,
and induction hypothesis yields some $k'_j$ such that $\Prog \sqchlrdcn{k'_j} \psi_j$;
or else $\psi_j$ is not defined, in which case $\Prog \sqchlrdcn{k'_j} \psi_j$ for $k'_j = 0$, by Lemma \ref{lem:ep}(\ref{lem:ep:1}).
In these conditions, $\Prog \sqchlrdcn{k'} \varphi$ holds for $k' = 1 + \Sigma_{j=1}^{m} k'_j$,
with a proof tree using a \textbf{SQDA} root inference based on $R_l$ instantiated by $\theta$. \mathproofbox
\end{description}
\end{proof*}

As an easy consequence of the previous theorem we get:

\begin{cor}[$\SQCHL(\simrel,\qdom,\cdom)$ is sound and complete]
\label{cor:correctness}
For any $\sqclp{\simrel}{\qdom}{\cdom}$-program $\Prog$ and any observable qc-atom $\varphi$, the following three statements are equivalent:
$$
    (a) ~ \Prog \sqchlrdc \varphi \hspace*{1cm}
    (b) ~ \Prog \model{\simrel,\qdom,\cdom} \varphi \hspace*{1cm}
    (c) ~ \Mp \isqchlrdc \varphi
$$
Moreover, we also have:
\begin{enumerate}
  \item {\em Soundness}: $\Prog \sqchlrdc \varphi \Longrightarrow \Prog \model{\simrel,\qdom,\cdom} \varphi$.
  \item {\em Completeness}: $\Prog \model{\simrel,\qdom,\cdom} \varphi \Longrightarrow \Prog \sqchlrdc \varphi$.
\end{enumerate}
\end{cor}
\begin{proof}
Soundness and completeness are just a trivial consequence of $(a) \Leftrightarrow (b)$.
To finish the proof it suffices to  prove the two equivalences $(a) \Leftrightarrow (c)$ and $(b) \Leftrightarrow (c)$.
This is done as follows:

\smallskip
\noindent [$(a) \Leftrightarrow (c)$] In the case that $\varphi$ is a defined qc-atom, $\Mp \isqchlrdc \varphi$ reduces to $\varphi \in \Mp$ which is equivalent to $\Prog \sqchlrdc \varphi$ by Theorem \ref{thm:SQCHL-leastmodel}.
 Otherwise, $\Prog \sqchlrdc \varphi \Longleftrightarrow \Mp \isqchlrdc \varphi$ holds because of Lemma \ref{lem:ep}(\ref{lem:ep:1}).

\smallskip
\noindent [$(b) \Rightarrow (c)$] Assume $\Prog \model{\simrel,\qdom,\cdom} \varphi$ and recall Definition \ref{defn:models}.
Then $\I \isqchlrdc \varphi$ for every qc-interpretation $\I$ such that $\I \model{\simrel,\qdom,\cdom} \Prog$.
In particular, $\Mp \isqchlrdc \varphi$, since $\Mp \model{\simrel,\qdom,\cdom} \Prog$ was proved in Theorem \ref{thm:tp-leastmodel}.

\smallskip
\noindent [$(c) \Rightarrow (b)$] Assume $\Mp \isqchlrdc \varphi$.
In order to obtain $\Prog \model{\simrel,\qdom,\cdom} \varphi$ we must prove:
$$(\star) \quad \I \isqchlrdc \varphi \mbox{ holds for any qc-interpretation } \I \mbox{ such that } \I \model{\simrel,\qdom,\cdom} \Prog \enspace .$$
In the case that $\varphi$ is a defined qc-atom, $\Mp \isqchlrdc \varphi$ reduces to $\varphi \in \Mp$,
which implies $(\star)$ because $\Mp$ is the least model of $\Prog$, as proved in Theorem \ref{thm:tp-leastmodel}.
In the case that $\varphi$ is not defined but  $\cdom$-based, $(\star)$ follows form the fact
that $\I \isqchlrdc \varphi$ holds for any arbitrary qc-interpretation $\I$, as proved in Lemma \ref{lem:ep}(\ref{lem:ep:1}).
\end{proof}

We close this subsection with a brief discussion on the relationship between the entailment
relation $\!\!\entail{\qdom,\cdom}\!\!$ used in this report and a different one that was proposed
in \cite{CRR08} and noted $\!\!\entail{\simrel,\qdom}\!\!$.
In contrast to $\!\!\entail{\qdom,\cdom}\!\!$, the entailment $\!\!\entail{\simrel,\qdom}\!\!$
depended on a given {\em similarity} relation $\simrel$.
In the context of the SQCLP scheme, one could think of an entailment 
$\!\!\entail{\simrel,\qdom,\cdom}\!\!$ depending on $\simrel$ and defined in the following way:
given two qc-atoms $\varphi$ and $\varphi'$, we could say that $\varphi$ $(\simrel,\qdom,\cdom)$-entails
$\varphi'$ (in symbols, $\varphi \entail{\simrel,\qdom,\cdom} \varphi'$) iff $\varphi : \cqat{A}{d}{\Pi}$ and
$\varphi' : \cqat{A'}{d'}{\Pi'}$ such that there is some substitution $\theta$ satisfying 
$\simrel(A',A\theta) = \lambda \neq \bt$, $d' \dleq \lambda$, $d' \dleq d$ and $\Pi' \model{\cdom} \Pi\theta$.

However, $\!\!\!\!\entail{\simrel,\qdom,\cdom}\!\!\!\!$ would not work properly in the case that $\simrel$ is not transitive, as shown by the following simple example: think of a $\sqclp{\simrel}{\U}{\rdom}$-program $\Prog$ including just a clause
\begin{itemize}
  \item[] $R_1 : p_1 \qgets{1.0}$
\end{itemize}
and assume that $\simrel$ verifies $\simrel(p_1,p_2) = 0.9$, $\simrel(p_2,p_3) = 0.9$ and $\simrel(p_1,p_3) = 0.4$ where $p_1,p_2,p_3 \in DP^0$. Then, $\Prog \sqchl{\simrel}{\U}{\rdom} \cqat{p_2}{0.9}{\emptyset}$ can be easily proved with the SQCHL rule {\bf SQDA} and $\cqat{p_2}{0.9}{\emptyset} \entail{\simrel,\U,\rdom} \cqat{p_3}{0.9}{\emptyset}$ holds because of $\simrel(p_2,p_3) = 0.9$, but $\Prog \sqchl{\simrel}{\U}{\rdom} \cqat{p_3}{0.9}{\emptyset}$ does not hold. Therefore, the Entailment Property for Programs (Lemma \ref{lem:ep}(\ref{lem:ep:2})) would fail if the entailment $\!\!\entail{\simrel,\qdom,\cdom}\!\!$ were adopted in place of $\!\!\entail{\qdom,\cdom}\!\!$.

Since the Entailment Property for Programs is a very natural condition that must be preserved, we conclude that
the entailment relation $\!\!\entail{\qdom,\cdom}\!\!$ used in this report is the right choice in a framework where the
underlaying proximity relation is not guaranteed to be a similarity.

%% file: J3_2.tex
\subsection{Goals and their Solutions}
\label{sec:sqclp:goals}

In this brief subsection we present  the syntax and declarative semantics of goals in the $\mbox{SQCLP}$ scheme,
and we define natural soundness and completeness properties which are expected to be fulfilled by goal solving devices.
These notions are intended as a useful tool to reason about the correctness of $\mbox{SQCLP}$ implementations
to be developed in the future.

In order to build goals for $\sqclp{\simrel}{\qdom}{\cdom}$-programs, we assume a countably infinite set $\War$ of so-called {\em qualification variables} $W$, disjoint from $\Var$ and $\cdom$'s signature $\Sigma$.
Goals for a given program $\Prog$ have the form 
$$G ~:~ \qat{A_1}{W_1},~ \ldots,~ \qat{A_m}{W_m} \sep W_1 \dgeq^? \!\beta_1,~ \ldots,~ W_m \dgeq^? \!\beta_m$$
abbreviated as $(\qat{A_i}{W_i},~ W_i \dgeq^? \!\beta_i)_{i = 1 \ldots m}$, where $\qat{A_i}{W_i} ~ (1 \leq i \leq m)$ are atoms annotated with different qualification variables $W_i$; and $W_i \dgeq^? \!\beta_i$ are so-called {\em threshold conditions} with $\beta_i \in \bqdom ~ (1 \leq i \leq m)$.
The notations ? and $\dgeq^?$ have been already explained in Section \ref{sec:sqclp:programs}.

In the sequel, the notation $\warset{o}$ will denote the set of all qualification variables occurring in the syntactic object $o$. In particular, for a goal $G$ as displayed above, $\warset{G}$ denotes the set $\{W_i \mid 1 \leq i \leq m\}$.
In the case $m = 1$ the goal is called  {\em atomic}.
The declarative semantics of goals is provided by their solutions, that are defined as follows: 

\begin{defn}[Goal Solutions]
\label{dfn:goalsol}
Assume a given $\sqclp{\simrel}{\qdom}{\cdom}$-program $\Prog$ 
and a goal $G$ for the program $\Prog$ with the syntax displayed above. 
Then:
\begin{enumerate}
\item 
A {\em solution} for $G$ is any triple $\langle \sigma, \mu, \Pi \rangle$ such that 
$\sigma$ is a $\cdom$-substitution, $\mu : \warset{G} \to \aqdom$, 
$\Pi$ is a satisfiable and finite set of atomic $\cdom$-constraints 
and the following two conditions hold for all $i = 1 \ldots m$:
  \begin{enumerate}
  \item
  $W_i\mu = d_i \dgeq^? \!\beta_i$ and
  \item
  $\Prog \sqchlrdc \cqat{A_i\sigma}{W_i\mu}{\Pi}$.
  \end{enumerate}
The set of all solutions for $G$ is noted $\Sol{\Prog}{G}$.
Note that solutions are {\em open} in the sense that the substitution $\sigma$ is not required to be ground.
\item
A solution $\langle \eta, \rho, \Pi \rangle$ for $G$ is called {\em ground} iff $\Pi = \emptyset$ and
$\eta \in \mbox{Val}_\cdom$ is a variable valuation such that $A_i\eta$ is a ground atom for all $i = 1 \ldots m$.
The set of all ground solutions for $G$ is noted $\GSol{\Prog}{G}$.
Obviously, $\GSol{\Prog}{G} \subseteq \Sol{\Prog}{G}$.  
\item
A ground solution $\langle \eta, \rho, \emptyset \rangle \in \GSol{\Prog}{G}$ is {\em subsumed} by 
$\langle \sigma, \mu, \Pi \rangle$ iff 
there is some $\nu \in \Solc{\Pi}$ s.t.
$\eta =_{\varset{G}} \sigma\nu$ and $W_i\rho \dleq W_i\mu$ for $i = 1 \ldots m$.
\mathproofbox
\end{enumerate}
\end{defn}

Implicitly, the first item in the previous definition requires $\cqat{A_i\sigma}{W_i\mu}{\Pi}$ 
to be observable qc-atoms in the sense of Definition \ref{defn:atoms-entail}, 
which is trivially true because $W_i\mu = d_i \in \aqdom$ and $\Pi$ is satisfiable. 
In fact, Definition \ref{defn:atoms-entail} was designed with the aim 
of using observable qc-atoms as observations of valid open solutions for atomic goals.
The next example illustrates the definition:

\begin{exmp}[Solutions for an atomic goals]
\label{exmp:goal-sol}
\begin{enumerate}
\item
 $G : \qat{\texttt{goodWork(X)}}{\texttt{W}}  \sep \texttt{W} \dgeq \texttt{(0.55,30)}$
is a goal for the program fragment  $\Prog$ shown  in Figure \ref{fig:sample},
and the arguments given near the beginning of Subsection \ref{sec:sqclp:programs} can be formalized to prove that 
$\langle \{\texttt{X} \mapsto \texttt{king\_liar}\}, \{\texttt{W} \mapsto \texttt{(0.6,5)}\}, \emptyset \rangle \in \mbox{Sol}_\Prog(G)$.
\item
As an additional example involving constraints, recall the $\sqclp{\simrel}{\U}{\rdom}$-program 
$\Prog$ presented in Example \ref{exmp:running}.
An atomic goal $G$ for this program is $\qat{p'(c'(Y),Z)}{W} \sep W {\geq^?} 0.75$.
Consider 
$\sigma = \{Z \mapsto c(X)\}$,
$\mu = \{W \mapsto 0.8\}$ and
$\Pi = \{cp_{>}(X,1.0), op_{+}(A,A,$ $X), op_{\times}(2.0,A,Y)\}$.
Note that $0.8 \geq 0.75$ and
$\Prog \vdash_{\simrel,\U,\rdom} \cqat{p'(c'(Y),Z)\sigma}{W\mu}{\Pi}$,
as we have seen in Example \ref{exmp:sqchl-inference}.
Therefore, the requirements of Definition \ref{dfn:goalsol} are fulfilled, and 
$\langle \sigma, \mu, \Pi \rangle \in \mbox{Sol}_\Prog(G)$. \enspace \mathproofbox
\end{enumerate}
\end{exmp}

In practice, users of $\mbox{SQCLP}$ languages will rely on  some available {\em goal solving system} 
for computing goal solutions.
The following definition specifies two important properties of goal solving systems:

\begin{defn}[Correct Goal Solving Systems]
\label{dfn:goalsolsys}
At a high abstraction level, a  {\em goal solving system} for $\sqclp{\simrel}{\qdom}{\cdom}$ 
can be thought as a device that takes a program $\Prog$ and a goal $G$ as input
and yields various triples $\langle \sigma, \mu, \Pi \rangle$,
called {\em computed answers}, as outputs.
Such a goal solving system is called:
\begin{enumerate}
 \item 
 {\em Sound} iff every computed answer  is a solution $\langle \sigma, \mu, \Pi \rangle \in \mbox{Sol}_\Prog(G)$.
 \item
 {\em Weakly complete} iff every ground solution $\langle \eta, \rho, \emptyset \rangle \in \mbox{GSol}_\Prog(G)$
 is subssumed by some computed answer.
 \item
 {\em Correct} iff it is both sound and weakly complete.
\enspace \mathproofbox
\end{enumerate}
\end{defn}

Every goal solving system for a SQCLP instance should be sound and ideally also weakly complete.
Implementing such systems is one of the major lines of future research mentioned in the Conclusions of this report. 

%% file: J4.tex
\section{Conclusions}
\label{sec:conclusions}


We have extended the classical CLP scheme to a new scheme SQCLP whose instances  $\sqclp{\simrel}{\qdom}{\cdom}$ are
parameterized by a proximity relation $\simrel$, a qualification domain $\qdom$ and a constraint domain $\cdom$.
In addition to the known features of CLP programming, the new scheme offers extra facilities for dealing with expert knowledge representation and flexible query answering.
Inspired by the observable CLP semantics in \cite{GL91,GDL95},
we have presented a declarative semantics for SQCLP that 
provides  fixpoint and proof-theoretical characterizations of least program models
as well as an implementation-independent  notion of goal solutions.

SQCLP is a quite general scheme. Different partial instantiations of its three parameters lead to 
more particular schemes, most of which can be placed in close correspondence to previous proposals.
The items below present seven particularizations, along with some comments which make use of 
the notions {\em threshold-free}, {\em attenuation-free} and {\em constraint-free} which have been 
explained at the beginning of Section \ref{sec:sqclp:programs}.

\begin{enumerate}
\item 
By definition, QCLP has instances $\mbox{QCLP}(\qdom,\cdom) \eqdef \mbox{SQCLP}(\sid,\qdom,\cdom)$,
where $\sid$ is the {\em identity} proximity relation.
The {\em quantitative} $\mbox{CLP}$ scheme proposed in \cite{Rie98phd} 
can be understood as a further particularization  of QCLP 
that works with threshold-free $\mbox{QCLP}(\U,\cdom)$ programs,
where $\U$ is the qualification domain of uncertainty values (see Subsection \ref{ssec:ud}).
%
\item 
By definition, SQLP has instances $\mbox{SQLP}(\simrel,\qdom) \eqdef \mbox{SQCLP}(\simrel,\qdom,\rdom)$,
where $\rdom$ is the real constraint domain (see Subsection \ref{sec:cdoms:domains}).
The scheme with the same name originally proposed in \cite{CRR08} 
can be understood as a restricted form of the present formulation;
it worked with  threshold-free and constraint-free $\mbox{SQLP}(\simrel,\qdom)$ programs
and it restricted the choice of the $\simrel$ parameter to transitive proximity (i.e. similarity)  relations. 
%
\item 
By definition, SCLP\footnote{Not to be confused with SCLP in the sense of \cite{BMR01}, discussed below.} has instances $\mbox{SCLP}(\simrel,\cdom)  \eqdef  \mbox{SQCLP}(\simrel,\B,\cdom)$,
where $\B$ is the qualification domain of classical boolean values  (see Subsection \ref{ssec:bd}).
Due to the fixed parameter choice $\qdom = \B$,
both attenuation values and threshold values become useless,  
and each choice of $\simrel$ must necessarily represent a crisp reflexive and symmetric relation.
Therefore, this new scheme is not so interesting from the viewpoint of uncertain and qualified reasoning.
\item 
By definition, QLP has instances $\mbox{QLP}(\qdom) \eqdef \mbox{SQCLP}(\sid,\qdom,\rdom)$.
The scheme with the same name originally proposed  in \cite{RR08} 
can be understood as a restricted form of the present formulation;
it worked with  threshold-free and constraint-free $\mbox{QLP}(\qdom)$ programs.
%
\item 
By definition, SLP has instances $\mbox{SLP}(\simrel) \eqdef  \mbox{SQCLP}(\simrel,\U,\rdom)$.
The pure fragment of \textsf{Bousi}\verb+~+\textsf{Prolog} \cite{JR09}
can be understood as a restricted form of SLP in the present formulation;
it works with threshold-free, attenuation-free and constraint-free $\mbox{SLP}(\simrel)$ programs.
Moreover, restricting the choice of $\simrel$ to similarity relations leads to 
SLP in the sense of \cite{Ses02} and related papers.
%
\item 
The CLP scheme can be defined  by  instances $\mbox{CLP}(\cdom) \eqdef \mbox{SQCLP}(\sid,\B,\cdom)$. 
Both attenuation values and threshold values are useless in CLP programs,  due to the fixed parameter choice $\qdom = \B$.
%
\item 
Finally, the pure LP paradigm can be defined as $\mbox{LP} \eqdef  \mbox{SQCLP}(\sid,\B,\hdom)$,
where $\hdom$ is the {\em Herbrand} constraint domain.
Again,  attenuation values and threshold values are useless in LP due to the fixed parameter choice $\qdom = \B$.
\end{enumerate}

In all the previous items, the schemes obtained by partial instantiation inherit the declarative semantics from SQCLP, using sets of observables of the form $\cqat{A}{d}{\Pi}$ as interpretations. 
A similar semantic approach were used in our previous papers \cite{RR08,CRR08}, except that $\Pi$ and equations were absent due to the lack of CLP features.
The other related works discussed in the Introduction view program interpretations 
as mappings $\I$ from the ground Herbrand base into some set of lattice elements (the real interval $[0,1]$ in many cases), as already discussed in the explanations following Definition \ref{defn:interpretations}.

As seen in Subsection \ref{sec:sqclp:goals}, SQCLP's semantics enables a declarative characterization of valid goal solutions. This fact is relevant for modeling the expected behavior of goal solving devices and reasoning about their correctness.
Moreover, the relations $\approx_{\lambda, \Pi}$ introduced for the first time in the present paper (see Definition \ref{defn:Pi-prox}) allow to specify the semantic role of $\simrel$ in a constraint-based framework, with less technical overhead than in previous related approaches.

A related work not mentioned in items 1--7 above is the semiring-based CLP of \cite{BMR01}, a scheme with instances SCLP(S) parameterized by a semiring $\mbox{S} = \langle A, +, \times, \mathbf{0}, \mathbf{1} \rangle$ whose elements are used to represent consistency levels in soft constraint solving. The semirings used in this approach can be equipped with a lattice structure whose {\em lub} operation is always $+$, but whose {\em glb} operation may be different from $\times$.
On the other hand, our qualification domains are defined as lattices with an additional attenuation operation $\circ$.
It turns out that the kind of semirings used in SCLP(S) correspond to qualification domains only in some cases.
Moreover, $\times$ is used in SCLP(S) to interpret logical conjunction in clause bodies and goals, while the {\em glb} operation is used in $\sqclp{\simrel}{\qdom}{\cdom}$ for the same purpose.
For this reason, even if $\qdom$ is ``equivalent'' to S, $\sqclp{\simrel}{\qdom}{\cdom}$ cannot be naturally used to express SCLP(S) in the case that $\times$ is not the {\em glb}.
Assuming that $\qdom$ is ``equivalent'' to S and that $\times$ behaves as the {\em glb} in S, program clauses in SCLP(S) can be viewed as a particular case of program clauses in $\sqclp{\simrel}{\qdom}{\cdom}$ which use an attenuation factor different from $\tp$ only for facts.
Other relevant differences between $\sqclp{\simrel}{\qdom}{\cdom}$ and SCLP(S) can be explained by comparing the parameters.
As said before $\qdom$ may be ``equivalent'' to S in some cases, but $\simrel$ is absent and $\cdom$ is not made explicit in SCLP(S).
Seemingly, the intended use of SCLP(S) is related to finite domain constraints and no parametrically given constraint domain is provided.


In the future we plan to implement some SQCLP instances by extending the semantically correct program transformation techniques from \cite{CRR08}, and to investigate applications which can profit from flexible query answering.
Other interesting lines of future work include: a) extension of the qualified SLD resolution presented in \cite{RR08} to a SQCLP goal solving procedure able to work with constraints and a proximity relation; and b) extension of the QCFLP scheme in \cite{CRR09} to work with a proximity relation and higher-order functions.